%% file: arxiv.tex
\documentclass{article} %
\usepackage{iclr2026_conference,times}

\input{math_commands.tex}

\usepackage{hyperref}
\usepackage{url}

\usepackage{amsthm}
\usepackage{amssymb}
\usepackage{booktabs}       %
\usepackage{graphicx}
\usepackage{mathtools}
\usepackage{caption}
\usepackage{tikz}
\usetikzlibrary{arrows.meta, positioning,calc}

\usepackage{algorithm}
\usepackage[noend]{algpseudocode}
\usepackage{caption}
\usepackage{titlesec}
\titlespacing*{\paragraph}{0pt}{0.5ex}{1.5ex}
  
\titlespacing*{\section}
  {0pt}   %
  {3pt}  %
  {-1pt}   %

\titlespacing*{\subsection}
  {0pt}   %
  {3pt}  %
  {-1pt}   %
  
\usepackage{enumitem} %
\setlist[enumerate,1]{left=0em,itemsep=0.0em,topsep=0.0em,parsep=0.4em}
\setlist[itemize,1]{left=0em,itemsep=0.0em,topsep=0.0em,parsep=0.4em}

\newcommand{\algstrut}[1][\algruledefaultfactor]{\vrule width 0pt depth .25\baselineskip
height #1\baselineskip\relax}
\newcommand*{\algrule}[1][\algorithmicindent]{\hspace*{.5em}\vrule\algstrut
\hspace*{\dimexpr#1-.5em}}

\makeatletter
\newcount\ALG@printindent@tempcnta
\def\ALG@printindent{%
\ifnum \theALG@nested>0%
\ifx\ALG@text\ALG@x@notext%

\else \unskip
\ALG@printindent@tempcnta=1 \loop \algrule[\csname ALG@ind@\the\ALG@printindent@tempcnta\endcsname]%
\advance \ALG@printindent@tempcnta 1 \ifnum \ALG@printindent@tempcnta<\numexpr\theALG@nested+1\relax%
\repeat \fi \fi }

\patchcmd{\ALG@doentity}{\noindent\hskip\ALG@tlm}{\ALG@printindent}{}{\errmessage{failed to patch}}

\AtBeginEnvironment{algorithmic}{\lineskip0pt}

\theoremstyle{plain}
\newtheorem{theorem}{Theorem}
\newtheorem{proposition}[theorem]{Proposition}

\theoremstyle{definition}
\newtheorem{definition}{Definition}
\newtheorem{assumption}[theorem]{Assumption}
\theoremstyle{remark}

\title{Watermarks for Language Models via \\Probabilistic Automata}

\author{
Yangkun Wang \\
University of California, San Diego \\
\texttt{yaw048@ucsd.edu} \\
\And
Jingbo Shang \\
University of California, San Diego \\
\texttt{jshang@ucsd.edu}
}

\iclrfinalcopy %

\begin{document}

\maketitle

\begin{abstract}
\input{0-abs}

\end{abstract}

\input{1-intro}

\input{2-related}
\input{3-problem}
\input{4-constr}

\input{5-undet}
\input{7-exp}

\input{8-concl}
\input{9-ack}

\bibliography{references}
\bibliographystyle{iclr2026_conference}

\newpage
\appendix
\input{900-limit}
\input{901-results}

\input{902-setup}
\input{903-details}

\input{904-analysis}

\input{905-proof}

\input{906-automaton}
\input{907-impact}

\end{document}

%% file: math_commands.tex
\usepackage{amsmath,amsfonts,bm}

\def\eqref#1{Equation~\ref{#1}}

\def\1{\bm{1}}

\DeclareMathAlphabet{\mathsfit}{\encodingdefault}{\sfdefault}{m}{sl}
\SetMathAlphabet{\mathsfit}{bold}{\encodingdefault}{\sfdefault}{bx}{n}

\newcommand{\E}{\mathbb{E}}

\usepackage{amsmath, bm}

\newcommand{\bs}{\mbox{\boldmath $s$}}

\newcommand{\bx}{\mbox{\boldmath $x$}}
\newcommand{\by}{\mbox{\boldmath $y$}}

\newcommand{\bbP}{{\mathbb{P}}}

\newcommand{\bbZ}{{\mathbb{Z}}}

\newcommand{\calA}{{\mathcal{A}}}

\newcommand{\calK}{{\mathcal{K}}}
\newcommand{\calL}{{\mathcal{L}}}
\newcommand{\calM}{{\mathcal{M}}}

\newcommand{\calO}{{\mathcal{O}}}

\newcommand{\calV}{{\mathcal{V}}}
\newcommand{\calW}{{\mathcal{W}}}

\newcommand{\bxi}{\mbox{\boldmath $\xi$}}

\newcommand{\bPhi}{\mbox{\boldmath $\Phi$}}

\DeclareMathOperator*{\Prob}{\mathbb{P}}

\newcommand{\sk}{\mathsf{sk}}
\newcommand{\indot}{\mathsf{IND}\text{-}\mathsf{OT}}
\newcommand{\indcpa}{\mathsf{IND}\text{-}\mathsf{CPA}}

%% file: 0-abs.tex
A recent watermarking scheme for language models achieves distortion-free embedding and robustness to edit-distance attacks. However, it suffers from limited generation diversity and high detection overhead.
In parallel, recent research has focused on undetectability, a property ensuring that watermarks remain difficult for adversaries to detect and spoof.
In this work, we introduce a new class of watermarking schemes constructed through \emph{probabilistic automata}.
We present two instantiations: (i) a practical scheme with exponential generation diversity and computational efficiency, and (ii) a theoretical construction with formal undetectability guarantees under cryptographic assumptions. Extensive experiments on LLaMA-3B and Mistral-7B validate the superior performance of our scheme in terms of robustness and efficiency.

%% file: 1-intro.tex
\section{Introduction}
The rapid development of large-scale language models (LMs) has markedly improved AI's ability to generate textual content \citep{brown2020language}. Despite these advancements, apprehensions have arisen over authenticity, ownership, and potential misuse of such technologies  \citep{zellers2019defending, solaiman2019release}. Traditional AI detection methods, such as classifier-based detection, often fall short in terms of robustness. In contrast, text watermarking offers a potential solution to these problems. It works by embedding a private key within the text that can be detected by the key holder, thereby identifying and minimizing the abuse of AI-generated content. 

A widely adopted watermarking method conditions the decoder on the preceding $k$ generated tokens \citep{kirchenbauer2023watermark,aaronson2022aisafety}. While effective, this approach can degrade LM's output quality by introducing noticeable distortions, such as biases toward certain $k$-grams. To address these distortions, distortion-free watermarking was introduced to preserve the LM's output distribution \citep{kuditipudi2024robust}; however, it does not guarantee LM's generation diversity. More recently, undetectable watermarking has been explored \citep{christ2024undetectable}, which prevents detection by adversaries and naturally maintains generation diversity. Despite these advancements, the relationship between distortion-freeness and undetectability has thus far never been clearly defined. To this end, we establish the connection between distortion-freeness and undetectability, and show that many existing watermarks are detectable. One interpretation of this fact is that the watermarking output distribution can be recognized by probabilistic deterministic finite automata (PDFA) and can therefore be learned under the Probably Approximately Correct (PAC) framework.

Our work is closely related to the state-of-the-art watermarking approach introduced by \citet{kuditipudi2024robust}, which uses a cyclic key sequence as noise for unbiased decoding and leverages the \textit{edit distance} (specifically, Levenshtein distance) metric to improve robustness against any edit-based attacks.
However, this method suffers from two notable drawbacks: (1) it reduces generative diversity, often leading to deterministic outputs, and (2) it requires partitioning text into blocks with the time complexity scales quadratically with the block size, which creates a major computational bottleneck.

We introduce a new class of watermarking schemes represented by \emph{probabilistic automata} (PA), with the following key contributions:

\begin{itemize}
  \item Our framework generalizes the cyclic key sequence watermarking of \citet{kuditipudi2024robust} as a special case, which can be modeled as a probabilistic deterministic finite automaton (PDFA) with a simple cyclic topology (see Figure~\ref{fig:seq-vs-pnfa}).

  \item We extend this formulation to \emph{probabilistic non-deterministic finite automata} (PNFA), a strictly more expressive class than PDFAs. Leveraging the fact that the class of languages recognized by PNFAs is not PAC-learnable under the sparse Learning Parity with Noise (LPN) assumption \citep{kearns1994learnability}, we construct an undetectable watermarking scheme.

  \item We instantiate a practical scheme that significantly improves both generation diversity and detection efficiency over \citet{kuditipudi2024robust}: (i) Increasing generation diversity from $\Theta(\lambda)$ to $\Omega(\lambda d^n)$, where $d \geq 1$, $\lambda$ is the key length, and $n$ is the sequence length, and (ii) Reducing detection time complexity from $\Theta(\lambda n k^2)$ to $\Theta(\lambda n)$, where $k$ is the block size.
\end{itemize}

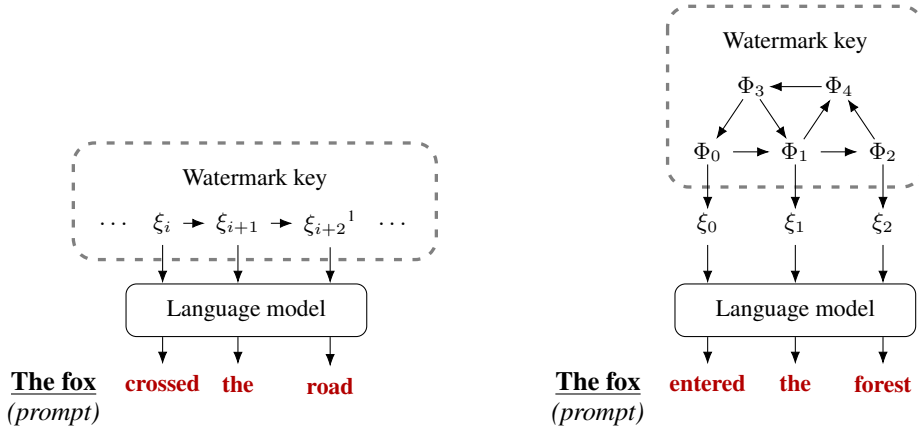
\begin{figure}[tbp]
    \centering
    \vspace{-2.4em}
    \begin{minipage}[t]{0.4\linewidth}
        \centering
        \input{tikz/tikz_cyclic}
    \end{minipage}%
    \hfill
    \begin{minipage}[t]{0.6\linewidth}
        \centering
        \input{tikz/tikz_wepa}
    \end{minipage}

    \caption{Comparison of generation of two watermarking schemes where the key follows a cyclic structure \citep{kuditipudi2024robust} on the left and a probabilistic automaton on the right. Each $\Phi$ specifies a probability distribution over $\xi$, with precise definitions provided in later sections.}
    \label{fig:seq-vs-pnfa}
    \vspace{-1.5em}
\end{figure}
\footnotetext{All subscripts are taken modulo $\lambda$.}

%% file: tikz/tikz_cyclic.tex
\begin{tikzpicture}[
  every node/.style={font=\small, inner sep=1pt},
  >=Latex,
  baseline=(prompt.base)
]

\node (xi0) at (0,0) {$\xi_{i}$};
\node (xi1) [right=0.5cm of xi0] {$\xi_{i+1}$};
\node (xi2) [right=0.5cm of xi1] {$\xi_{i+2}\footnotemark$};
\node (xi3) [right=0.2cm of xi2] {$\dots$};
\node (xim1) [left=0.2cm of xi0] {$\dots$};

\node (key) [above=0.4cm of $(xim1)!0.5!(xi3)$] {Watermark key};

\coordinate (rect-nw) at ([xshift=-0.3cm, yshift=0.3cm]xim1.west |- key.north);
\coordinate (rect-se) at ([xshift=0.3cm, yshift=-0.3cm]xi3.east |- xi0.south);
\draw[dashed, dash pattern=on 3pt off 4pt, line width=1.25pt, rounded corners=8pt, draw=gray] (rect-nw) rectangle (rect-se);

\node (lm) [
    draw,
    rounded corners,
    below=0.8cm of $(xi0)!0.5!(xi2)$,
    minimum width=3.2cm,
    minimum height=0.7cm
] {Language model};

\node (y0) [below=1.8cm of xi0, text=red!70!black] {\textbf{crossed}};
\node (y1) [below=1.8cm of xi1, text=red!70!black] {\textbf{the}};
\node (y2) [below=1.8cm of xi2, text=red!70!black] {\textbf{road}};

\node (prompt) [left=0.3cm of y0, font=\bfseries] {\underline{The fox}};
\node (promptlabel) [below=0cm of prompt, font=\itshape] {(prompt)};

\foreach \i in {0,1} {
  \pgfmathtruncatemacro{\j}{\i + 1}
  \draw[->, shorten >=3pt, shorten <=3pt] (xi\i) -- (xi\j);
}

\foreach \i in {0,1,2} {
  \draw[->, shorten >=0pt, shorten <=3pt] (xi\i) -- (lm.north -| xi\i);
  \draw[->, shorten >=3pt, shorten <=0pt] (lm.south  -| xi\i) -- (y\i);
}
\end{tikzpicture}

%% file: tikz/tikz_wepa.tex
\begin{tikzpicture}[
  every node/.style={font=\small, inner sep=1pt},
  >=Latex,
  baseline=(prompt.base)
]

\node (phi0) at (0,0) {$\Phi_0$};
\node (phi1) [right=0.7cm of phi0] {$\Phi_1$};
\node (phi2) [right=0.7cm of phi1] {$\Phi_2$};
\node (phi3) [above=0.7cm of $(phi0)!0.5!(phi1)$] {$\Phi_3$};
\node (phi4) [above=0.7cm of $(phi1)!0.5!(phi2)$] {$\Phi_4$};

\node (xi0) [below=0.6cm of phi0] {$\xi_0$};
\node (xi1) [below=0.6cm of phi1] {$\xi_1$};
\node (xi2) [below=0.6cm of phi2] {$\xi_2$};

\node (key) [above=0.4cm of $(phi3)!0.5!(phi4)$] {Watermark key};

\coordinate (rect-nw) at ([xshift=-0.3cm, yshift=0.3cm]phi0.west |- key.north);
\coordinate (rect-se) at ([xshift=0.3cm, yshift=-0.3cm]phi2.east |- phi2.south);
\draw[dashed, dash pattern=on 3pt off 4pt, line width=1.25pt, rounded corners=8pt, draw=gray] (rect-nw) rectangle (rect-se);

\node (lm) [
    draw,
    rounded corners,
    below=0.8cm of $(xi0)!0.5!(xi2)$,
    minimum width=3.2cm,
    minimum height=0.7cm
] {Language model};

\node (y0) [below=1.8cm of xi0, text=red!70!black] {\textbf{entered}};
\node (y1) [below=1.8cm of xi1, text=red!70!black] {\textbf{the}};
\node (y2) [below=1.8cm of xi2, text=red!70!black] {\textbf{forest}};

\node (prompt) [left=0.3cm of y0, font=\bfseries] {\underline{The fox}};
\node (promptlabel) [below=0cm of prompt, font=\itshape] {(prompt)};

\draw[->] (phi0) -- (xi0);
\draw[->] (phi1) -- (xi1);
\draw[->] (phi2) -- (xi2);

\foreach \i in {0,1} {
  \pgfmathtruncatemacro{\j}{\i + 1}
  \draw[->, shorten >=3pt, shorten <=3pt] (phi\i) -- (phi\j);
}
\draw[->] (phi3) -- (phi0);
\draw[->] (phi3) -- (phi1);
\draw[->] (phi4) -- (phi3);
\draw[->] (phi1) -- (phi4);
\draw[->] (phi2) -- (phi4);

\foreach \i in {0,1,2} {
  \draw[->, shorten >=0pt, shorten <=3pt] (xi\i) -- (lm.north -| xi\i);
  \draw[->, shorten >=3pt, shorten <=0pt] (lm.south  -| xi\i) -- (y\i);
}

\end{tikzpicture}

%% file: 2-related.tex
\section{Related work}

\paragraph{Watermarking of language models.}
Text watermarking aims at blending a private ``key'' in text generation so that it can be detected by key holders.
Early approaches relied on subtle text modifications using heuristics~\citet{atallah2001natural,atallah2002natural,topkara2005natural}.
Following the autoregressive nature of language models (LMs), recent watermarking methods start to condition the token generation by the key and $k$ prior tokens~\citep{kirchenbauer2023watermark,aaronson2022aisafety,zhao2024provable}.
However, these methods can significantly alter the underlying LM's distribution, for instance, by introducing biases for certain $k$-grams.

Watermarking without changing the next token distribution of the LM on a \emph{single text sample} is defined as \emph{distortion-free} watermarking~\citep{kuditipudi2024robust,hu2024unbiased}.
For example, \citet{kuditipudi2024robust} uses a cyclic key sequence of noise for unbiased decoding and was the first to employ the \textit{edit distance} metric as the alignment between text and the key sequence for detection.
While this method enhances the robustness of watermarking against edit-based attacks, it suffers from two drawbacks of lacking generation diversity, and it relies on partitioning the text into blocks and repeatedly shifting the key sequence to compute edit distances multiple times, posing significant efficiency challenges.

Recent works have constructed \emph{undetectable} watermarks theoretically that require watermarked texts to be indistinguishable across \emph{multiple queries}, yet none of them remains practical. For instance, \citet{christ2024undetectable}'s construction is based on hash functions and not robust to edit-based attacks.
\citet{christ2024pseudorandom} assume a binary symmetric channel model for LMs, which is clearly unrealistic. \citet{golowich2024edit} make an assumption that the vocabulary size scales polynomially with the security parameter (i.e., the size of the ``key''), which does not typically hold in practice. Notably, all of these undetectable watermarking schemes rely on the construction of pseudorandom functions, yet they are based on disparate assumptions and lack a unified framework.

\paragraph{Probabilistic automata.}
Probabilistic automata (PA) are widely studied in computational linguistics that describe distributions with latent variables over finite sequences of symbols.
The class of PA consists of probabilistic nondeterministic finite automata (PNFA) and their proper subclass, probabilistic deterministic finite automata (PDFA). The learnability of these automata has gained profound theoretical interest and practical relevance, particularly in modeling distributions over strings.

\citet{kearns1994learnability} explored the complexity of learning PDFAs within the Probably Approximately Correct (PAC) framework. Specifically, they established that PAC-learning PDFAs with a two-letter vocabulary is at least as hard as PAC-learning noisy parities, which is believed to be computationally hard.
Consequently, the entire class of PDFAs cannot be PAC-learned within polynomial time constraints.
On the other hand, \citet{ron1995learnability,clark2004pac} showed that under the constraints of a certain distinguishability on the states, acyclic PDFA and cyclic PDFAs with bounded expected string length from any states are PAC-learnable.

Although there have been works on learning PDFAs and intermediate forms of PNFAs, few studies address the problem of learning general PNFAs due to their hardness. \citet{terwijn2002learnability} demonstrated that PNFAs are not PAC-learnable if Blum integer factorization is hard. \citet{angluin1991won} showed that learning remains hard even if adversaries have oracle access to membership queries (in the context of watermarking the queries of the detection function).

\paragraph{Edit distance and error correction.}
Edit distance quantifies the similarity between two strings by counting the minimum number of operations required to transform one into the other. This concept extends to formal languages, where the edit distance is defined as the minimum distance between any pair of sequences within the languages. This metric is particularly valuable in error correction, where it helps identify the closest valid string to a given input string. \citet{wagner1974order} introduced an error-correcting algorithm that constructs a finite state automaton to recognize a set of strings.
Their approach uses dynamic programming to compute the edit distance between a string and a regular language.

%% file: 3-problem.tex
\section{Preliminaries}
We denote the alphabet of an automaton by $\Sigma$ and the vocabulary of a language model by $\calV$ to explicitly distinguish them. A language is defined as a mapping $\Sigma^* \to [0, 1]$ (or $\calV^* \to [0, 1]$, respectively). Given an automaton $\mathcal{M}$, we denote by $\calL(\mathcal{M}) \subseteq \Sigma^*$ the language recognized by $\mathcal{M}$, i.e., the set of strings accepted by $\mathcal{M}$. We denote the size of the watermark key by $\lambda$. In undetectable watermarking settings, it also serves as the security parameter, which measures the strength of distortion-freeness and undetectability of a watermarking scheme. Further discussion on $\lambda$ will follow in Section~\ref{sec:distortion}. The length of the generated sequence is given by $m=\calO(\text{poly}(\lambda))$, where $\text{poly}(\cdot)$ denotes a polynomial function. For a sequence $\bx = (x_1, \dots, x_n)$, we write $\bx_{i:} = (x_i, \dots, x_n)$ for a suffix, and $\bx_{i:j} = (x_i, \dots, x_j)$ for a contiguous subsequence.

\begin{definition}[Language Model]
    An (autoregressive) language model is defined by a function $\mathsf{Model}:
    \calV^{*}\to \Delta(\calV)$ that maps a sequence of tokens to a probability distribution over $\calV$, where $\Delta(\cdot)$ denotes a probability distribution over a set. Given an initial sequence of tokens (a \textit{prompt}) $\bx \in \calV^*$, the probability of a sequence $\by = (y_{1}, y_{2}, \ldots, y_{m})$  is defined as
    \begin{equation}
        \label{eq:language-model}p(\by) = \prod_{i=1}^{m}p(y_{i}\mid \bx, \by_{1:i-1}),
    \end{equation}
    where $\by_{l:h}= (y_{l}, y_{l+1}, \ldots, y_{h})$ represents a subsequence of $\by$. Each conditional probability is modeled by
    \begin{equation}
        p(\cdot \mid \bx, \by_{1:i-1}) = \mathsf{Model}(\bx, \by_{1:i-1}).
    \end{equation}
    We use the notation $\by \xleftarrow{\mathsf{AR}}\mathsf{Model}(\bx)$ to indicate that $\by$ is autoregressively generated by $\mathsf{Model}$ given $\bx$, following~\eqref{eq:language-model}.
\end{definition}

\begin{definition}[Decoder-based Watermarking Scheme]
    A decoder-based watermarking scheme is a triplet $\mathcal{W}:= (\mathsf{Gen}, \mathsf{Model}^{\mathsf{wat}}, \mathsf{Detect})$ such that:
    \begin{enumerate}
        \item The key generation algorithm $\mathsf{Gen}$ is randomized and takes
            as input $1^{\lambda}$ to generate a secret key $\sk \in \calK$:
            \footnote{The input is given in unary notation to ensure polynomial runtime in $\lambda$. The structure of the key can be arbitrary, with specifics described in later sections.}
            \begin{equation}
                \sk \leftarrow \mathsf{Gen}(1^{\lambda}).
            \end{equation}

        \item The watermarking algorithm $\mathsf{Model}_{\sk}^{\mathsf{wat}}:= (\mathsf{Model}, \Phi, \Gamma)$ is an autoregressive process that consists of an unwatermarked model $\mathsf{Model}$, a noise generator $\Phi : \calK \times \Xi^{*}\times \calV^{*}\to \Delta(\Xi)$, and a decoder $\Gamma : \Xi \times \Delta(\calV) \to \calV$, where $\Xi$ is the noise space.
        At the $i$-th decoding step, $\xi_{i}\in \Xi$ is sampled by
        \begin{equation}
            \label{eq:s-to-xi}
            \xi_{i}\sim \Phi_{\sk}(\bxi_{1:i-1},\bx, \by_{1:i-1}),\end{equation}
        then the next token $y_{i}$ is produced by the decoder
        deterministically given the noise $\xi_i$:
        \begin{equation}
            \label{eq:decoder}
            y_{i}\leftarrow \Gamma(\xi_{i}, \mathsf{Model}(
            \bx, \by_{1:i-1})).
        \end{equation}

        \item The detection algorithm $\mathsf{Detect}_{\sk}$ takes as input $\sk$
            and $\by$, and outputs $\mathsf{true}$ or $\mathsf{false}$.
    \end{enumerate}
\end{definition}
While the above definition presented seems abstract, it covers a wide range of existing watermarking frameworks. Specific examples are provided in Appendix~\ref{sec:analysis}.

Ideally, $\mathsf{Detect}_{\sk}(\by)$ should output $\mathsf{true}$ if $\by$ is generated by $\mathsf{Model}_{\sk}^{\mathsf{wat}}(\bx)$ for some $\bx$, and output $\mathsf{false}$ if $\by$ is independent of $\sk$. The former property is referred to as \textit{completeness} and the latter \textit{soundness}.

As a special case, when $\Phi$ does not depend on the prefix of tokens $\bx$ and
$\by_{1:i-1}$, ~\eqref{eq:s-to-xi} simplifies to
\begin{equation}
\label{eq:agnostic}
    \xi_{i}\sim \Phi_{\sk}(\bxi_{1:i-1}).
\end{equation}
We refer to this case as \textit{model-agnostic} as the distribution of the noise can be decomposed autoregressively with the chain rule and does not depend on the specific model used and therefore can be precomputed before decoding.

All decoder-based watermarking schemes require sufficiently high text entropy; otherwise, the output tends to be deterministic and no watermarks can be embedded. For a detailed discussion of entropy effects, refer to \citet{christ2024undetectable,kuditipudi2024robust}; we do not repeat these efforts here. In the sequel a watermarking scheme always denotes a decoder-based watermarking scheme.

The robustness of a watermarking scheme quantifies its ability to withstand edit-based corruptions to the watermarked data without losing the embedded watermark.

\begin{definition}[Edit Distance]
    The edit distance $d(\bs_{1}, \bs_{2})$ between two sequences $\bs_{1}, \bs_{2}
    \in \calV^{*}$ is the minimum cost of transforming $\bs_{1}$ into $\bs_{2}$
    through a sequence of single-position edit operations, including insertion, deletion,
    and substitution.
\end{definition}

\begin{definition}[Robustness]
    A watermarking scheme is considered robust if, for any watermarked sequence
    $\by$ from $\mathsf{Model}_{\sk}^{\mathsf{wat}}(\bx)$ and any sequence
    $\by'$ with the edit distance bounded by
    $d(\by, \by') \leq \gamma\max(|\by|,|\by'|)$ for some $\gamma > 0$, the detection
    function reliably identifies the watermark:
    \begin{equation}
        \mathsf{Detect}_{\sk}(\by') = \mathsf{true}.
    \end{equation}
\end{definition}

Note that a sequence detected as watermarked is not necessarily generated by the watermarked model. Robustness ensures that a watermarked sequence and its close neighbors remain detectable, but this does not compromise the property of soundness.

%% file: 4-constr.tex
\section{Constructing watermarks through probabilistic automata}
\label{sec:construction}

We begin by introducing the relevant definitions that underlie our watermarking constructions.

\begin{definition}[Probabilistic Non-Deterministic Finite Automaton]
A probabilistic non-deterministic finite automaton (PNFA) defined as a tuple $(Q, \Sigma, \delta, \pi_{0}, \pi_{f})$, where
(1) $Q$ is a finite set of states;
(2) $\Sigma$ is a finite alphabet of input symbols;
(3) $\delta: Q \times \Sigma \times Q \to [0, 1]$ is the transition probability function;
(4) $\pi_{0}: Q \to [0, 1]$ defines the initial probability of each state;
(5) $\pi_{f}: Q \to [0, 1]$ defines the final probability of each state.
\end{definition}

\begin{definition}[Probabilistic Deterministic Finite Automaton]
A probabilistic non-deterministic finite automaton (PNFA) $(Q, \Sigma, \delta , \pi_{0}, \pi_{f})$ is a probabilistic deterministic finite automaton (PDFA) if:
(1) $\exists q_{0}\in Q$ such that $\pi_{0}(q_{0}) = 1$ and $\forall q \in Q \setminus \{q_{0}\},$ $\pi_{0}(q) = 0$, and
(2) $\forall q \in Q, \forall a \in \Sigma$, there exists at most one state $q' \in Q$ such that $\delta(q, a, q') > 0$.
\end{definition}

\subsection{Watermarking schemes represented by probabilistic automata}
\label{subsec:pa}
In model-agnostic watermarking schemes, given a secret key $\sk$, $\Phi_{\sk}$ defines a distribution of the random variable $\bxi \in \Xi^*$ by applying \eqref{eq:agnostic} autoregressively. One case of this distribution is the stochastic language recognized by a PA. We model this distribution using a hierarchical automaton. Specifically, for a PA $\calM = (Q, \Sigma, \delta, \pi_{0}, \pi_{f})$ with $\Sigma \subset \Delta(\Xi)$, the process starts from an initial state $q_0$. Each transition produces $\Phi_i \in \Sigma$, a probability distribution over $\Xi$, from which noise $\xi_i \sim \Phi_i$ is sampled. The noise $\xi_i$, represented by a binary sequence, is used for decoding the next token $y_i$ as described in~\eqref{eq:decoder}. Each probability distribution $\Phi_i$ is modeled by a subordinate PA with a binary alphabet. The hierarchical structure allows the PA to model the noise distribution represented by binary alphabet.

\subsection{Constructing watermarks for language models}
\label{sec:constr-lm}
We now elaborate a specific construction of the watermarking scheme. We consider a decoder that uses exponential minimum sampling following \citet{aaronson2022aisafety,kuditipudi2024robust}. The decoder generates the next token based on a noise $\xi$ and the model's output probabilities, which can be formally expressed as
\begin{equation}
    \label{eq:gumbel-decoder}\Gamma(\xi_{i}, \mathsf{Model}(\by_{1:i-1})) = \mathop
    {\arg\min}_{j \in \calV}(\pi_{j} / \log(\mu_{j})),
\end{equation}
where $\pi_{j}$ is the probability assigned by $\mathsf{Model}(\by_{1:i-1})$ to token $j \in \calV$, and $\xi_i = (\mu_1, \dots, \mu_{|\calV|})$ with $\mu_j \overset{\text{i.i.d.}}{\sim} \text{Uniform}[0,1]$. This decoder preserves the model’s categorical distribution at each step.

Transitioning from continuous to binary representation, any real number $z \in [0, 1)$ can be approximated using its binary expansion:
\begin{equation}
    \label{eq:uni-to-bin}z = \frac{1}{2^{c}}\sum_{i=0}^{c-1}2^{i}\sigma_i, \quad \sigma_i\in\{0,1\},
\end{equation}
where $c$ denotes the precision level, and $\sigma_{i}$ represents the $i$-th bit of the binary expansion of $z$.

We introduce a PA to generate a sequence of noise. The PA consists of $\lambda$ virtual states, each corresponding to a subordinate PA that models a specific noise distribution. These states are labeled as $q_{0},q_{1},\dots,q_{\lambda-1}$ with each state $q_{i}$ transitioning to $q_{i+1 \bmod \lambda}, \dots, q_{i+d \bmod \lambda}$ with equal probability, thereby forming a $d$-regular graph. The arrangement of these states is strategically designed for robustness and efficiency.

The subordinate probabilistic automaton is defined over a vocabulary of size $|\mathcal{V}|$, with bitwidth $b$ and $c \geq b$. At each decoding step, it generates a binary noise vector $\xi = (\mu_1, \dots, \mu_{|\mathcal{V}|})$, where each $\mu_i \in \{0,1\}^c$.

The automaton begins at an initial state $q_0$ and terminates at a final state $q_f$, progressing through $|\mathcal{V}|$ layers that each encode a binary vector $\mu_i$. The first layer starts with: $q_0 \to \sigma_{1,1}$, and each layer proceeds through intermediate bitwise states: $\sigma_{i,j} \to \sigma_{i,j+1}, \text{for } 1 \le i \le |\mathcal{V}|,\ 1 \le j < b$, where $\sigma_{i,j}$ encodes the $j$-th bit of $\mu_i$. At $\sigma_{i,b}$, the automaton branches into two parallel Boolean paths: $\sigma_{i,b} \to \iota_{i,b+1}, \sigma_{i,b} \to \hat{\iota}_{i,b+1}$, which continue as: 
$\iota_{i,j} \to \iota_{i,j+1},\ 
\iota_{i,j} \to \hat{\iota}_{i,j+1},\ 
\hat{\iota}_{i,j} \to \iota_{i,j+1},\ 
\hat{\iota}_{i,j} \to \hat{\iota}_{i,j+1}$,
where $\iota_{i,j} = 0$ and $\hat{\iota}_{i,j} = 1$ represent bit encodings of $\mu_i$. Between layers, transitions connect the terminal states of layer $i$ to the initial states of layer $i{+}1$:
$
\label{eq:interlayer}
\iota_{i,c} \to \sigma_{i+1,1},
\hat{\iota}_{i,c} \to \sigma_{i+1,1},
\text{for } b \le j < c,
$
and the automaton concludes after the final layer with $\iota_{|\mathcal{V}|,c}, \hat{\iota}_{|\mathcal{V}|,c} \to q_f$.

For each state, all outgoing transitions have equal probability. As a special case where $d=1$ and sufficiently large $b$ and $c$, the PA produces noise equivalent to the cyclic key sequence watermarking \citep{kuditipudi2024robust}. The PA sampling process is integrated into the token generation, as described in Algorithm~\ref{algo:watermarking}. An illustration of a subordinate PA is provided in Appendix~\ref{app:subpa}.

We now proceed to describe the detection algorithm. We begin by defining a cost following \citet{aaronson2022aisafety} and \citet{kuditipudi2024robust} as
\begin{equation}
    d_{0}(y, \xi) = \log(1 - \mu_{y}),
\end{equation}
where $\xi = (\mu_{1}, \mu_{2}, \dots, \mu_{|\calV|})$ and $\mu_i \in [0,1]$ is represented by its binary expansion using the subordinate PA. For $\Phi \in \Delta(\Xi)$, the cost is defined as
\begin{equation}
\label{eq:d0_dist}
    d_{0}(y, \Phi) = \max_{\xi \in \operatorname{supp}(\Phi)}\{d_0(y,\xi)\}.
\end{equation}
Before we define the Levenshtein distance for PAs, we introduce the necessary definition.

\begin{definition}[Support Automaton]
    The support automaton of a PA $\calM = (Q, \Sigma, \delta, \pi_{0}, \pi_{f})$ is a non-deterministic finite automaton (NFA) $\underline\calM =(Q, \Sigma, \delta, Q_{0}, Q_{f})$, where $Q_{0}$ (respectively $Q_{f}$) denotes the set of initial (respectively final) states, and $\delta \subseteq Q \times \Sigma \times Q$ is the transition function defined as $(q, a, q') \in \delta \Leftrightarrow \phi(q, a, q') > 0.$
\end{definition}

\begin{algorithm}[htbp]
    \caption{Watermarking Schemes Represented by PAs}
    \label{algo:watermarking}
    \begin{algorithmic}[1]
        \Require $\mathsf{Model}$, PA $\calM$, decoder $\Gamma$
        \Ensure Watermarked sequence $\by$
        \State Sample an initial state $q$ from $\calM$.
        \For{each decoding step $i$}
            \State Transition $q$ to the next state and obtain $\Phi_i \in \Delta(\Xi)$.
            \State Sample $\xi_{i} \sim \Phi_i$.
            \State Select the next token $y_{i}$ using $\xi_{i}$ via~\eqref{eq:gumbel-decoder}.
        \EndFor
        \State \Return $\by$
    \end{algorithmic}
\end{algorithm}

\begin{definition}[Generalized Levenshtein Distance]
    For a sequence $\by \in \calV^{*}, \bPhi \in \Delta^*(\Xi), \gamma_i,\gamma_d$, the Levenshtein distance between $\by$ and $\bPhi$ is defined recursively as
    \begin{equation}
        \label{eq:gerneralized-levenshtein}
        d_{L}(\by, \bPhi) = \min
        \begin{cases}
            \gamma_{d} |\by| + \gamma_i |\bPhi|, \text{if }|\by| = 0 \text{ or }|\bPhi| = 0, \\
            d_{L}(\by_{2:}, \bPhi) + \gamma_d,                                        \\
            d_{L}(\by, \bPhi_{2:}) + \gamma_i,                                        \\
            d_{L}(\by_{2:}, \bPhi_{2:}) + d_{0}(y_1, \Phi_1).
        \end{cases}
    \end{equation}
    For a PA $\calM$ defining a support language $\calL({\underline\calM}) \subseteq \Delta^*(\Xi)$, the Levenshtein distance between $\by$ and $\calM$ is defined as
    \begin{equation}
        \label{eq:gerneralized-levenshtein2}
        d_{L}(\by, \calM) = \min_{\bPhi \in \calL({\underline\calM})}\left\{d_{L}(\by ,\bPhi
        )\right\}.
    \end{equation}
\end{definition}

In~\eqref{eq:gerneralized-levenshtein}, the behavior of $\by$ and $\bPhi$ is not symmetric, as the length of $\bPhi$ can be infinite. Therefore, we assign different costs to insertion  ($\gamma_i$) and deletion ($\gamma_d$). Meanwhile, the generalized Levenshtein distance is different from the design in \citet{kuditipudi2024robust}. In their algorithm, the lengths of the key sequence and the text are constrained to be equal, and it may disrupt their alignment when insertions or deletions are involved. In contrast, our definition allows flexible alignment despite difference in lengths.

The naive dynamic programming computes~\eqref{eq:gerneralized-levenshtein2} in $\calO(dm\lambda\log \lambda))$ time \citep{wagner1974order}, whereas our PNFA construction reduces it to $\calO(m\lambda)$. The dynamic programming algorithm is detailed in Appendix~\ref{app:edit-dist}.

\paragraph{Empirical $p$‑value.}
We quantify the alignment between an observed sequence $\by$ and the watermarked model $\mathcal{M}_{\sk}$ by the statistic $\psi \;=\; d_{L}\bigl(\by,\mathcal{M}_{\sk}\bigr)\,,$
and treat as null hypothesis $H_{0}$ that $\by$ is independent of~$\sk$. We sample $N$ independent keys $\{\sk_{i}\}_{i=1}^{N}$ and compute $\psi_{i}\;=\;d_{L}\bigl(\by,\mathcal{M}_{\sk_{i}}\bigr)
\quad (i=1,\dots,N).$
The empirical $p$‑value is then
\begin{equation}
\label{eq:p-value}
\hat p \;=\;\frac{1 + \sum_{i=1}^{N}\mathbf{1}\bigl[\psi_{i}\le\psi\bigr]}{N + 1}.
\end{equation}
A small $p$-value indicates that $\by$ aligns unusually well with the true key (suggesting a watermark), whereas a larger $p$-value is consistent with unwatermarked noise.

\paragraph{An upper bound of $p$-value via Vysochanskij–Petunin Inequality.}
Let $\mu$ and $\sigma^{2}$ denote the sample mean and variance of $\{\psi_{i}\}_{i=1}^{N}$, given by $\mu = \frac{1}{N}\sum_{i=1}^{N}\psi_{i}$ and $\sigma^{2}= \frac{1}{N-1}\sum_{i=1}^{N}(\psi_{i}- \overline \psi)^{2}$, and define the standardized statistic
\begin{equation}
\label{eq:z-score}
    z = (\psi - \mu)/\sigma.
\end{equation}
Under the mild assumption that the null distribution of $d_{L}$ is unimodal, the one‑sided Vysochanskij–Petunin inequality \citep{mercadier2021one} gives the following upper bound:
\begin{equation}
    \Pr[Z \le z] \le
    \begin{dcases}
        4/9(z^2+1)              & \text{if }|z| \ge \sqrt{5/3}, \\
        4/3(z^2+1)- 1/3 & \text{otherwise}.
    \end{dcases}
\end{equation}
We refer to this scheme as WEPA, for Watermarking schemE through Probabilistic Automata. The design of WEPA makes it robust and highly efficient for dynamic programming. As observed, disregarding bitwidth constraints, the generation diversity of WEPA is $\Theta(\lambda d^m)$, which is a substantial improvement over the construction of \citet{kuditipudi2024robust}.

%% file: 5-undet.tex
\section{Expressiveness of PA-based watermarking schemes}
\label{sec:distortion}
In this section, we draw connections between distortion-freeness, undetectability, and security levels in watermarking schemes. The security level is measured by the security parameter $\lambda$ that determines the key size and reflects the computational hardness of breaking cryptographic schemes. It is assumed to be known to adversaries, and their running time is expressed as functions of $\lambda$ rather than concrete values. An adversary is considered efficient if it can be modeled as a probabilistic algorithm with running time bounded by a polynomial function of $\lambda$. A function is called \textit{negligible} if it becomes asymptotically smaller than the inverse of any polynomial as $\lambda$ grows, which represents a probability that is practically negligible for large $\lambda$.
Negligible functions can be expressed as
$\text{negl}(\lambda)=\calO(\frac{1}{\text{poly}(\lambda)})$ (i.e., functions that vanish faster than any inverse polynomial).

\paragraph{Distortion-freeness.}
We begin the definition by considering an adversary who has oracle access to the unwatermarked model (i.e., can make unlimited queries) and observes a single instance of watermarked data, and its goal is to determine whether the observed data is watermarked or not. While we make assumptions about the adversary's computational capabilities, we make no assumptions about its strategy, which means the definition provides protection against any computationally bounded adversary. The idea is that the adversary should be unable to extract any partial information that distinguishes watermarked data from unwatermarked data.

In addition, we do not make assumptions about the adversary's prompting strategy. The adversary can manipulate the model's behavior by prompting so that the model outputs arbitrary distribution. This ensures that the security definition remains robust against adversarial control over the model’s output distribution.
Since the adversary can shape the output distribution arbitrarily, the specific input to the model becomes irrelevant---choosing a prompt equates to choosing a model. Therefore, we can safely omit the model’s input in our analysis.

The definition is formalized using the indistinguishability against one-time attacks ($\indot$) game, defined for a watermarking scheme $\calW$, an adversary $\calA$, and the security parameter $\lambda$:
(1) A secret key is generated via $\sk \leftarrow \mathsf{Gen}(1^{\lambda})$;
(2) $\mathcal{A}$ is given an input $1^{\lambda}$ and chooses $\mathsf{Model}
(\cdot)$; \footnote{We omit the input $\bx$ to $\mathsf{Model}$ as choosing a prompt equates to choosing a model as explained.}
(3) a uniform bit $b \in \{0, 1\}$ is selected. A sequence $\by \xleftarrow{\mathsf{AR}}f(\cdot)$ is computed and given to $\calA$, where $f = \mathsf{Model}$ if $b = 0$ otherwise $f = \mathsf{Model}_{\sk}^{\mathsf{wat}}$;
(4) finally, $\mathcal{A}$ outputs a bit $b^{\prime}$.
The advantage of $\mathcal{A}$ in this game is defined as
\begin{equation}
  \text{Adv}_{\calW}^{\indot}(\calA) := \left| \Pr \left[ b=b' \right] - \frac{1}{2}
  \right|.
\end{equation}

\begin{definition}[Distortion-freeness]
  \label{def:distortionfreeness} A watermarking scheme $\mathcal{W}$ is \textit{distortion-free}
  if it is indistinguishable against one-time attacks. Formally, for any
  polynomial-time $\mathcal{A}$,
  \begin{equation}
    \label{eq:distortionfreeness}\text{Adv}_{\calW}^{\indot}(\calA) \le \text{negl}
    (\lambda).
  \end{equation}
  Moreover, a watermarking scheme is considered \textit{perfectly distortion-free}
  if $\text{Adv}^{\indot}_{\mathcal{W}}(\mathcal{A})=0$.
\end{definition}

The following theorem provides an equivalent definition of perfect distortion-freeness.

\begin{theorem}
  \label{thm:perfect-distortion-free} A watermarking scheme is perfectly distortion-free
  iff for every language model that defines language distribution $\calL$,
  every $\by \in \calL$, it follows that
  \begin{equation}
    \label{eq:perfect-distortion-free}\Pr[Y = \by] = \E_{\sk \leftarrow
    \textnormal{\textsf{Gen}}(1^\lambda)}\big [\Pr [ Y =\by \mid K = \sk] \big],
  \end{equation}
  where $K$ is the random variable of the secret key.
\end{theorem}

We defer this proof to Appendix~\ref{app:perfect-distortion-free}. Theorem~\ref{thm:perfect-distortion-free} indicates that the randomness in generating an output from the unwatermarked model can be equivalently represented by first sampling a key from the key space and then generating the sequence conditioned on that key. Consequently, the original distribution of the model's outputs is preserved. However, this condition is often too strict, and even \citet{kuditipudi2024robust}'s construction is not perfectly distortion-free if the text length is greater than $\lambda$. Instead, our definition emphasizes \textit{computational indistinguishability}, which is more attainable and aligns with the limitations of real-world systems.

\paragraph{Undetectability.}
Distortion-freeness ensures that the quality of the language model is not changed by preserving the distribution in one query. It is nonetheless detectable with multiple queries, and a higher level of security is desired. Consider the following indistinguishability against chosen prompt attack ($\indcpa$) game, defined for a watermarking scheme $\calW$, an adversary $\calA$, and the security parameter $\lambda$:
(1) A secret key is generated via $\sk \leftarrow \mathsf{Gen}(1^{\lambda})$;
(2) $\mathcal{A}$ is given an input $1^{\lambda}$ and the oracle access to $\mathsf{Model'}_{\sk}^{\mathsf{wat}}(\cdot)$ for any $\mathsf{Model'}$, and chooses $\mathsf{Model}(\cdot)$;
(3) a uniform bit $b \in \{0, 1\}$ is selected. A sequence $\by \xleftarrow{\mathsf{AR}}f(\cdot)$ is computed and given to $\calA$, where $f = \mathsf{Model}$ if $b = 0$ otherwise $f = \mathsf{Model}_{\sk}^{\mathsf{wat}}$;
(4) finally, $\mathcal{A}$ outputs a bit $b^{\prime}$.
The advantage of $\mathcal{A}$ in this game is defined as
\begin{equation}
  \text{Adv}_{\calW}^{\indcpa}(\calA) := \left| \Pr \left[ b=b' \right] - \frac{1}{2}
  \right|.
\end{equation}

\begin{definition}[Undetectability]
  \label{def:undetectability} A watermarking scheme $\mathcal{W}$ is \textit{undetectable}
  if, for all polynomial-time adversaries, there exists a negligible function
  $\text{negl}$ such that
  \begin{equation}
    \label{eq:undetectability}\text{Adv}_{\calW}^{\indcpa}(\calA) \le \text{negl}
    (\lambda).
  \end{equation}
\end{definition}

By definition it is clear that all undetectable watermarking schemes are distortion-free due to the stronger capabilities of adversaries. This implies that undetectable watermarking maintains the output distribution of the watermarked model unchanged. Conversely, any watermarking scheme that introduces distortion is also detectable.
This further implies that the watermarks in \citet{kirchenbauer2023watermark,zhao2024provable,aaronson2022aisafety,kuditipudi2024robust} are all detectable.

The following theorem highlights the \textit{expressiveness} (i.e., the extent to which the automaton represents various sequences) of PAs for constructing watermarking schemes.

\begin{theorem}
\label{thm:pnfa-undetectable}
(Abridged) Under certain conditions, there exists an undetectable watermarking scheme that can be represented by a PA if sparse LPN is hard.
\end{theorem}

See Appendix~\ref{app:pnfa-undetectable} for the full statement and proof. With this framework, the cyclic key sequence watermarking scheme \citep{kuditipudi2024robust} can be represented by a PDFA with $\lambda$ states and cyclic transitions, though its expressiveness is limited by its simple structure. To achieve higher expressiveness, we consider using PNFAs, which are strictly more expressive than PDFAs.

%% file: 7-exp.tex
\section{Empirical evaluation}
\label{sec:exp}

We conducted an empirical evaluation of WEPA’s statistical properties and robustness on LLaMA-3.2-3B \citep{dubey2024llama} and Mistral-7B (v0.3) \citep{jiang2023mistral} language models. In line with prior works \citep{kirchenbauer2023watermark,kuditipudi2024robust}, we generated watermarked text continuations from the news-like subset of the C4 dataset \citep{raffel2020exploring}. We compared WEPA with three baselines: cyclic key sequence watermarking (EXP) with exponential minimum sampling \citep{kuditipudi2024robust}, unigram green-red set watermarking (KGW) \citep{kirchenbauer2023watermark}, and unbiased watermarking (Unbiased) \citep{hu2024unbiased}. Notably, the KGW baseline is not directly comparable, as it may introduce noticeable distortion. For WEPA and EXP, we computed $p$-values via~\eqref{eq:p-value} with a sample size of 10,000 for consistency. For Unbiased, we report the upper bound of $p$-values as described in \cite{hu2024unbiased}. All methods were evaluated using their strongest hyperparameter settings as recommended by the original authors; details are provided in Appendix~\ref{app:hypermeters}. For each experiment, we report the $1/3$, $1/2$ (median), and $2/3$ quantiles of $p$-values for watermarked text across $200$ samples. Further empirical results are included in Appendix~\ref{app:exp}.

\subsection{Varying text lengths}
\label{sec:varying-text-lengths}
We varied the generation length of watermarked text from $4$ to $20$ tokens, as shown in Figure~\ref{fig:varying-lengths}. The results highlight differences between the two language models due to variations in their generation entropy. Among the methods, WEPA ($d=1$) and EXP consistently achieve the strongest detection performance. WEPA ($d=2$) performs slightly worse, followed by KGW. The Unbiased method lags behind, especially on shorter sequences. We further report ROC-AUC and TPR@1\%FPR in Appendix~\ref{app:varying-length-aucroc}. Notably, watermarking short text is inherently more challenging than long text. For completeness, we also include experiments on longer sequences (up to 200 tokens) in Appendix~\ref{app:vary-length-long}.

\begin{figure*}[t]
  \centering
    \vspace{-2.2em} %
  \begin{minipage}[b]{0.67\linewidth}
    \centering
    \includegraphics[width=\linewidth]{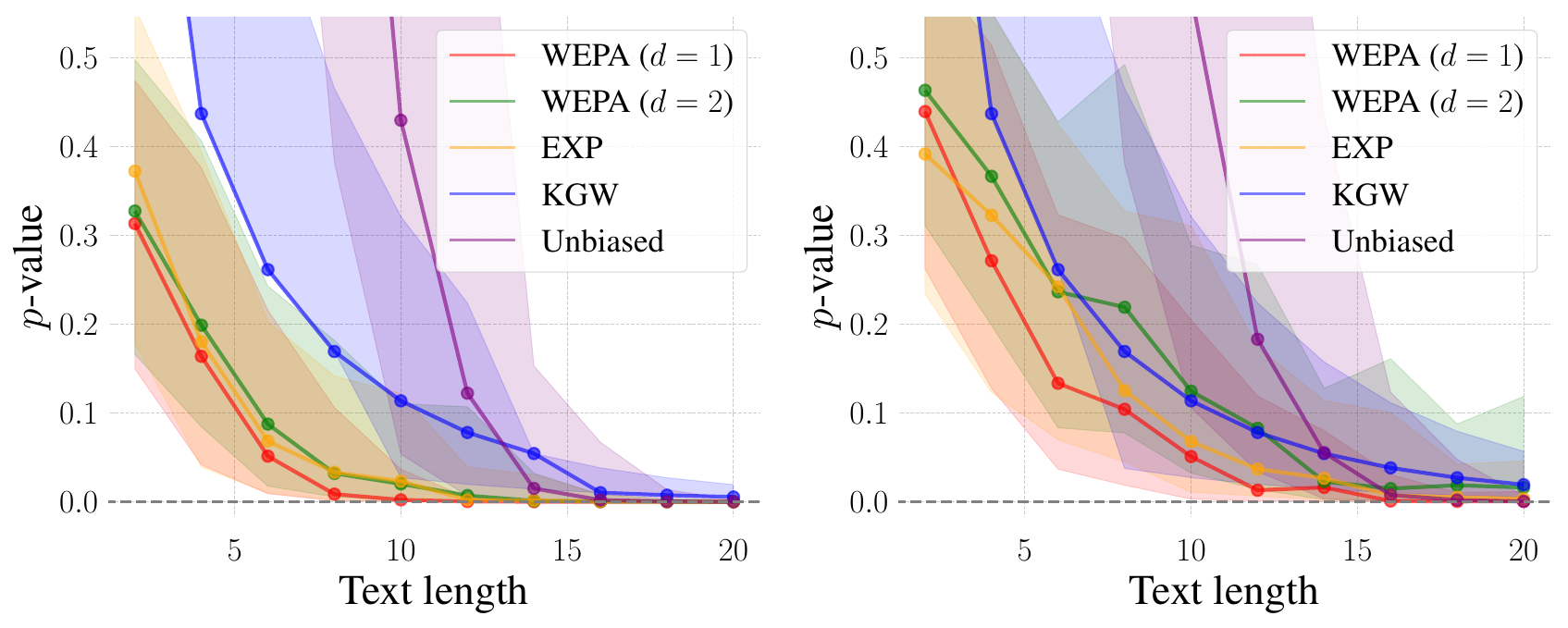}
    \vspace{-1em} %
    \captionof{figure}{Median $p$-values varying text lengths on LLaMA-3B (left) and Mistral-7B (right).}
    \label{fig:varying-lengths}
  \end{minipage}%
  \hfill
  \begin{minipage}[b]{0.31\linewidth}
    \centering
    {\footnotesize
    \setlength{\tabcolsep}{3pt}
    \begin{tabular}{@{}lccc@{}}
      \toprule
      \textbf{Scheme} & \textbf{Time} (s) \\
      \midrule
      WEPA ($d=1$) & 7.81 \\
      WEPA ($d=2$) & 8.22 \\
      WEPA ($d=1,b=6$) & 10.60 \\
      \midrule
      EXP (our impl.) & 2039.56 \\
      EXP & 49024.57 \\
      \bottomrule
    \end{tabular}
    }
    \captionof{table}{Detection efficiency comparison. Unless otherwise specified, \texttt{float32} is used for $b$ in WEPA.}
    \label{tab:efficiency}
    \end{minipage}
    \vspace{-0.8em}
\end{figure*}

\begin{figure}[tbp]
\vspace{-0.2em}
    \centering
    \includegraphics[width=\linewidth]{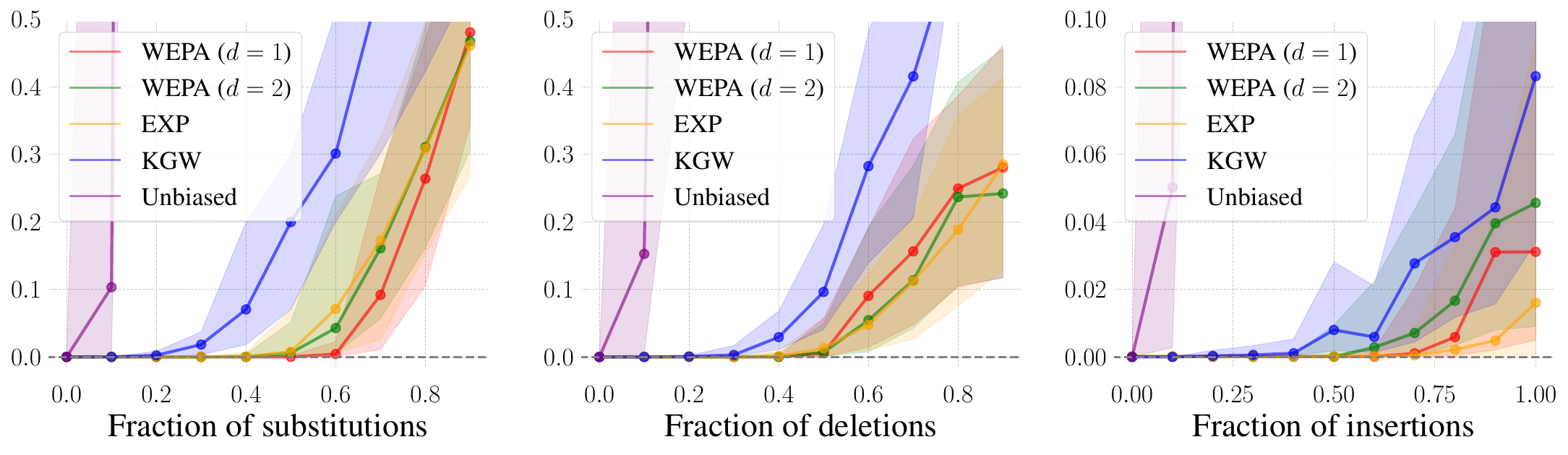}
    \caption{Median $p$-values under substitution (left), deletion (middle), and insertion (right) attacks.}
    \label{fig:corruption}
    \vspace{-1.5em}
\end{figure}

\subsection{Robustness to edit-based attacks}
\label{sec:robust}
We evaluated the robustness against three types of edit-based attacks: (1) substitution, (2) deletion, and (3) insertion, each applied by randomly corrupting a fraction of tokens. For substitution and insertion, replacement tokens were sampled uniformly from the vocabulary. No padding was introduced, nor was truncation applied following deletions or insertions. All experiments were conducted on sequences of length $50$ tokens.

Figure~\ref{fig:corruption} reports results on LLaMA-3B, while we defer the results on Mistral-7B in Appendix~\ref{app:corruption2}. WEPA ($d=1$) performs slightly better than the other methods across all attack types, likely due to its flexible alignment mechanism. In contrast, EXP performs slightly worse, as it requires exact alignment between the generated text and a fixed-length key sequence. KGW shows reduced robustness, likely due to its unigram-based design. The Unbiased method is not robust under these attacks, as it was not designed to handle edit-based perturbations.

We note that our watermark is specifically designed to resist edit-based perturbations, and does not aim to defend against semantic or paraphrasing attacks. As such, we do not include experiments on paraphrasing, which fall outside the scope of our work.

\subsection{Efficiency}

We compared the runtime of the detection algorithms for WEPA ($d=1$), WEPA ($d=2$) and EXP. For WEPA ($d=2,b=6$), we choose $b=6$ as the bitwidth does not affect efficiency. The results are presented in Table~\ref{tab:efficiency}.
For EXP, we implemented an optimization with token discretization. Each algorithm was evaluated on a sample from C4 dataset with text length of $256$ using LLaMA's tokenizer. WEPA is significantly faster than EXPEXP as predicted by the different time complexities.

%% file: 8-concl.tex
\section{Conclusion}
\label{sec:conclusion}
We introduced a class of watermarking schemes constructed through probabilistic automata. Within this framework, we instantiated WEPA, a practical watermarking method that achieves both improved generation diversity and more efficient detection. Empirical results further demonstrate its effectiveness and efficiency. Furthermore, by extending to probabilistic non-deterministic finite automata, we established an undetectable watermarking scheme.

%% file: 9-ack.tex
\section*{Acknowledgements}
Our work is sponsored in part by NSF CAREER Award 2239440, NSF Proto-OKN Award 2333790, Sponsored Research Projects from companies like Cisco and eBay, as well as generous gifts from Google, Adobe, and Teradata. Any opinions, findings, and conclusions or recommendations expressed herein are those of the authors and should not be interpreted as necessarily representing the views, either expressed or implied, of the U.S. Government. The U.S. Government is authorized to reproduce and distribute reprints for government purposes not withstanding any copyright annotation hereon.

%% file: 900-limit.tex
\section{Limitations}
\label{sec:limitations}
Our work has several limitations that suggest directions for future research. First, the lack of tight analytical bounds on the $p$-value remains a challenge, as the edit distance metric is inherently difficult to analyze. Second, the detection algorithm requires knowledge of the private key, which limits detection to only the key holder. Developing an asymmetric decoding and detection system could improve security.

%% file: 901-results.tex
\section{Additional empirical results}
\label{app:exp}

\subsection{Varying text lengths on long texts}
\label{app:vary-length-long}
We evaluated WEPA and EXP on longer text sequences, as presented in Figure~\ref{fig:varying_length_long}. For cases where the empirical $p$-values became exponentially small, direct computation was infeasible; instead, we report upper bounds estimated via \eqref{eq:z-score}. Due to the looseness of this bound, these values are not directly comparable to KGW. Nonetheless, the results reveal consistent trends with those in Figure~\ref{fig:varying-lengths}. Notably, WEPA ($d=1$) achieves the strongest performance, while WEPA ($d=2$) exhibits slightly lower accuracy compared to EXP.
\begin{figure}[htbp]
\centering
\includegraphics[width=\linewidth]{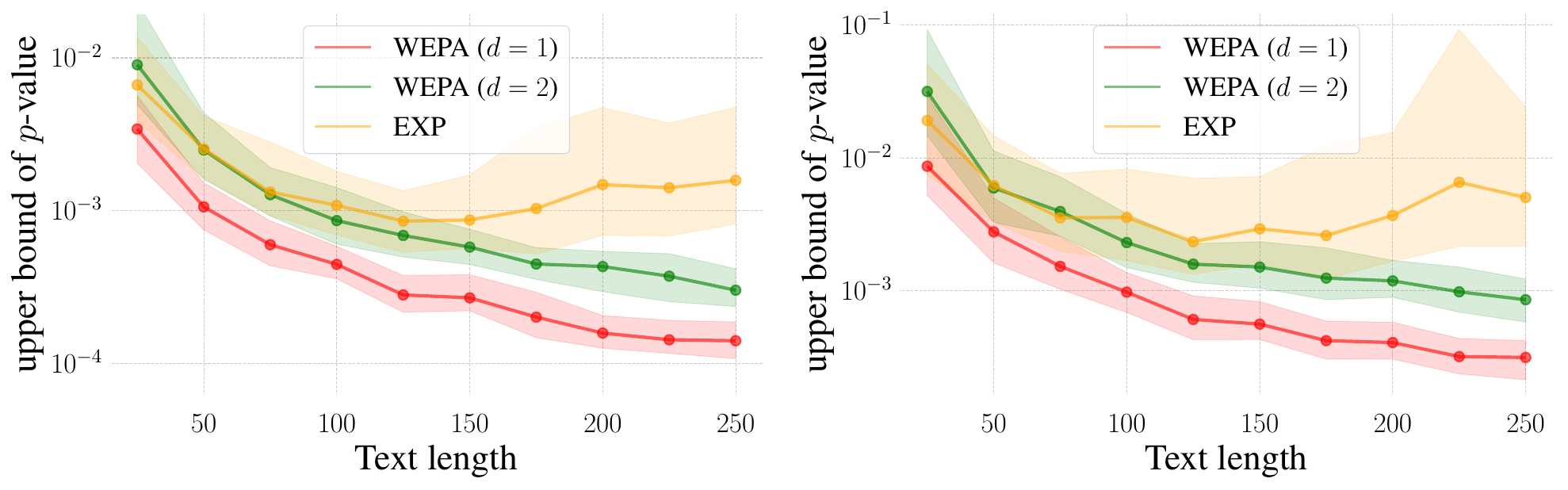}
\captionof{figure}{Median $p$-values varying text lengths on LLaMA-3B (left) and Mistral-7B (right).}
\label{fig:varying_length_long}
\end{figure}

\subsection{Perplexity evaluation}
Figure \ref{fig:ppl} presents the perplexity distributions across watermarking methods. KGW noticeably alters the text, while WEPA and Unbiased maintain perplexity levels comparable to unwatermarked outputs.

\begin{figure}[t]
    \centering
    \includegraphics[width=0.8\linewidth]{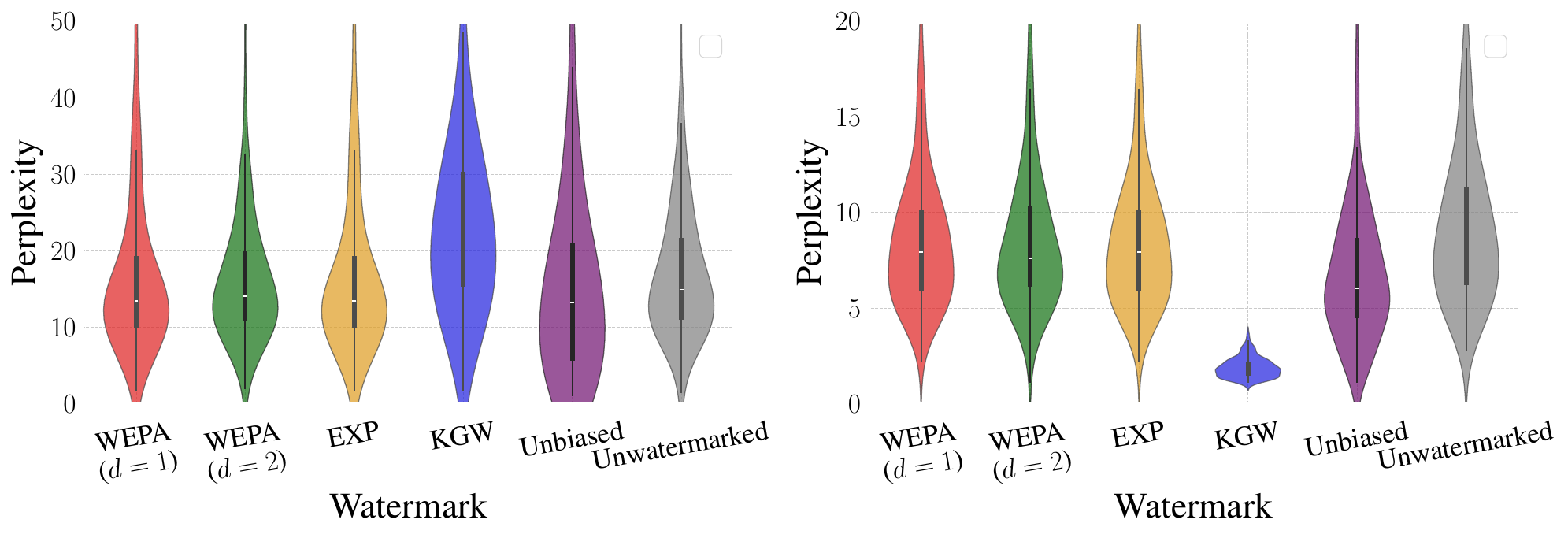}
    \caption{Perplexity distribution of each watermarking scheme on LLaMA-3B (left) and Mistral-7B (right).}
    \label{fig:ppl}
\end{figure}

\subsection{Varying text lengths with supplemental evaluation metrics}
\label{app:varying-length-aucroc}
Figure~\ref{fig:varying-text-length-aucroc} presents ROC-AUC results across varying text lengths, while Figure~\ref{fig:varying-text-length-tpr-at-fpr} reports the corresponding TPR@1\%FPR results. Overall, WEPA ($d=1$) achieves performance comparable to EXP on both metrics, with WEPA ($d=2$) performing even better. Although the Unbiased method attains the highest scores on these metrics, it is not robust to edit-based attacks.

\begin{figure}[htbp]
\centering
\includegraphics[width=\linewidth]{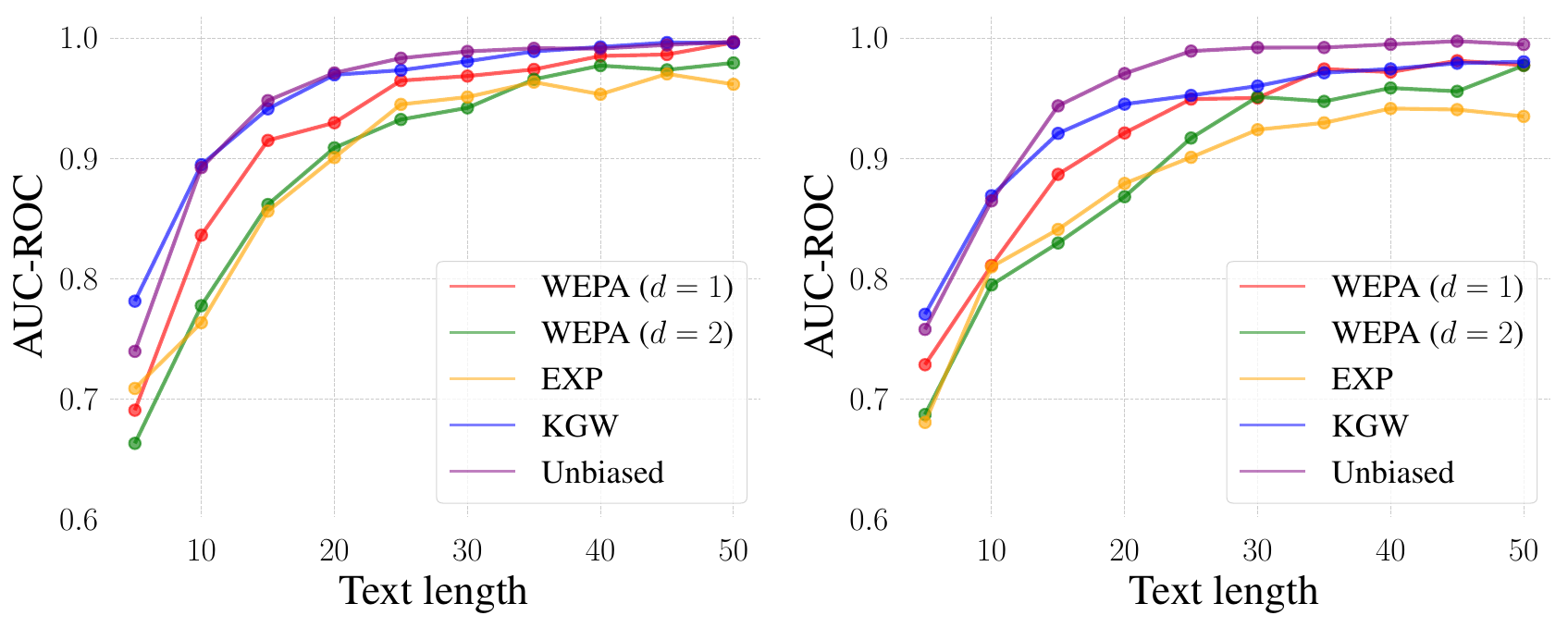}
\captionof{figure}{ROC-AUC results varying text lengths on LLaMA-3B (left) and Mistral-7B (right).}
\label{fig:varying-text-length-aucroc}
\end{figure}

\begin{figure}[htbp]
\centering
\includegraphics[width=\linewidth]{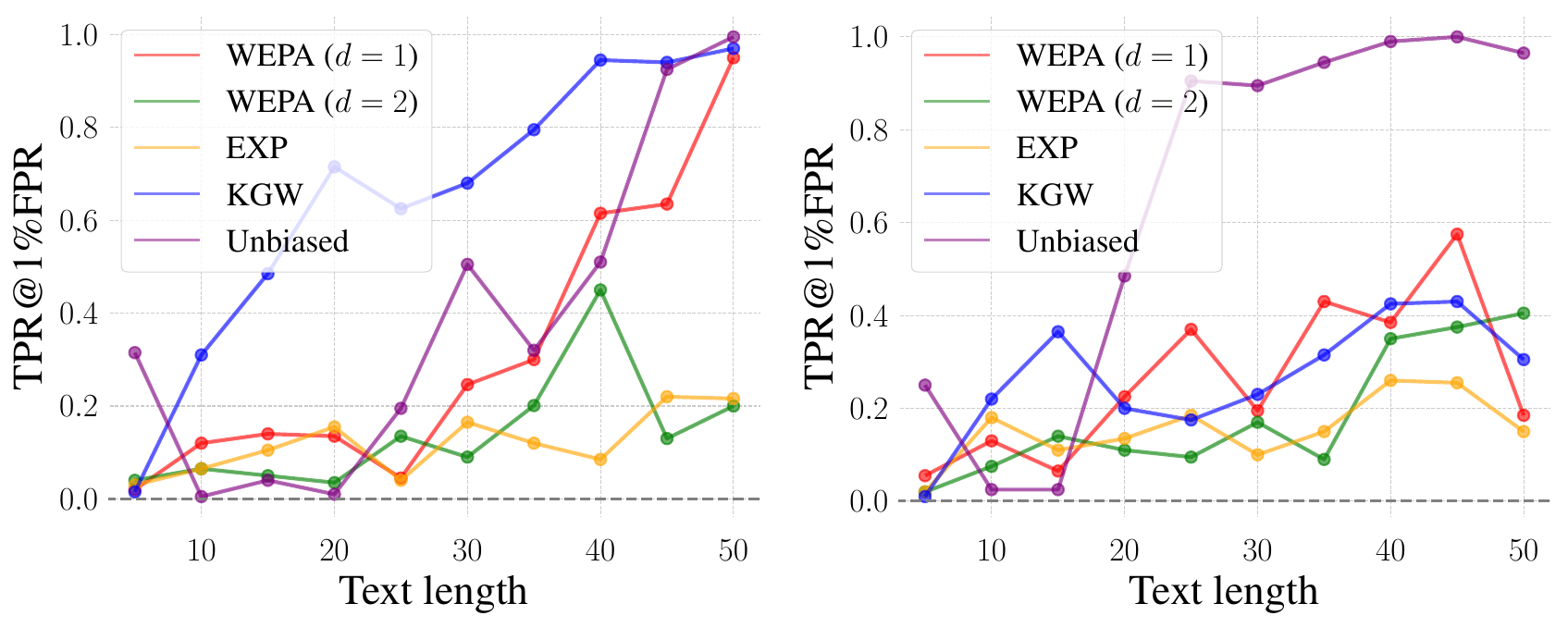}
\captionof{figure}{TPR@1\%FPR results varying text lengths on LLaMA-3B (left) and Mistral-7B (right).}
\label{fig:varying-text-length-tpr-at-fpr}
\end{figure}

\subsection{Varying key length}
We varied the key length ($\lambda$) from $2^4$ to $2^{12}$ while keeping the text length fixed at $8$, as shown in Figure~\ref{fig:varying_lambda}. Consistent with the findings of \citet{kuditipudi2024robust}, we observe that the $p$-value increases with $\lambda$.

\begin{figure}[htbp]
\centering
\includegraphics[width=\linewidth]{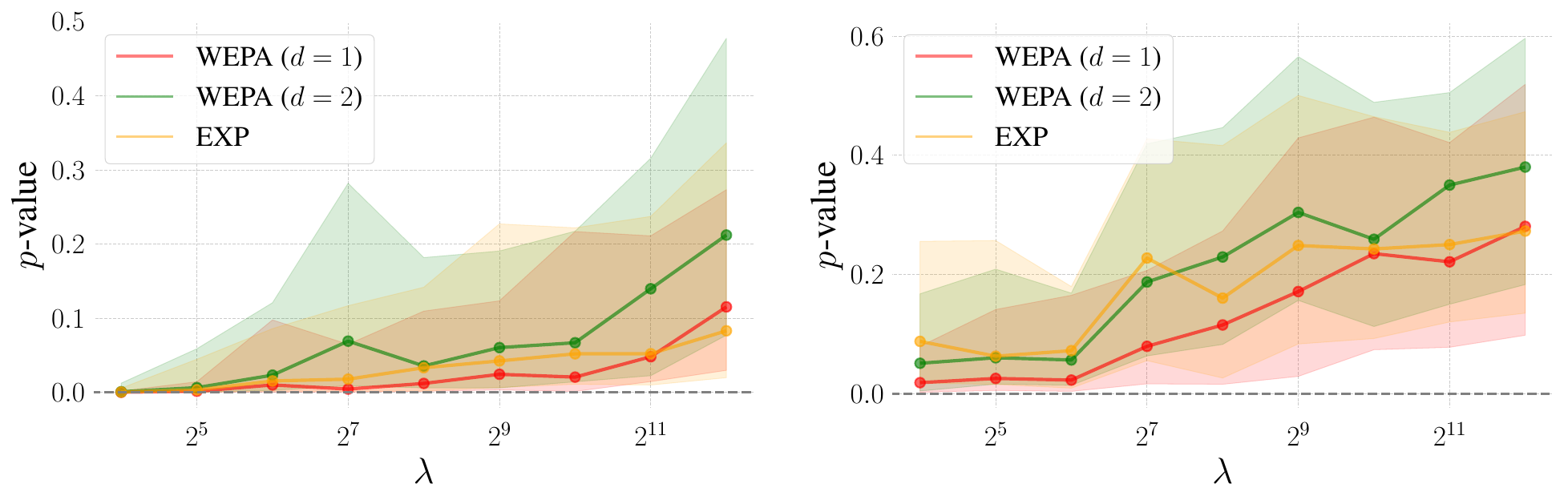}
\captionof{figure}{Median $p$-values varying $\lambda$s on LLaMA-3B (left) and Mistral-7B (right).}
\label{fig:varying_lambda}
\end{figure}

\subsection{Varying bitwidth}
We varied the bitwidth ($b$) of WEPA from $2$ to $12$ while keeping the text length fixed at $8$, as shown in Figure~\ref{fig:varying_bits}. The results indicate that a higher bitwidth improves performance with sacrifice of text diversity. The performance stabilizes when $b \geq 6$.
\begin{figure}[htbp]
\centering
\includegraphics[width=\linewidth]{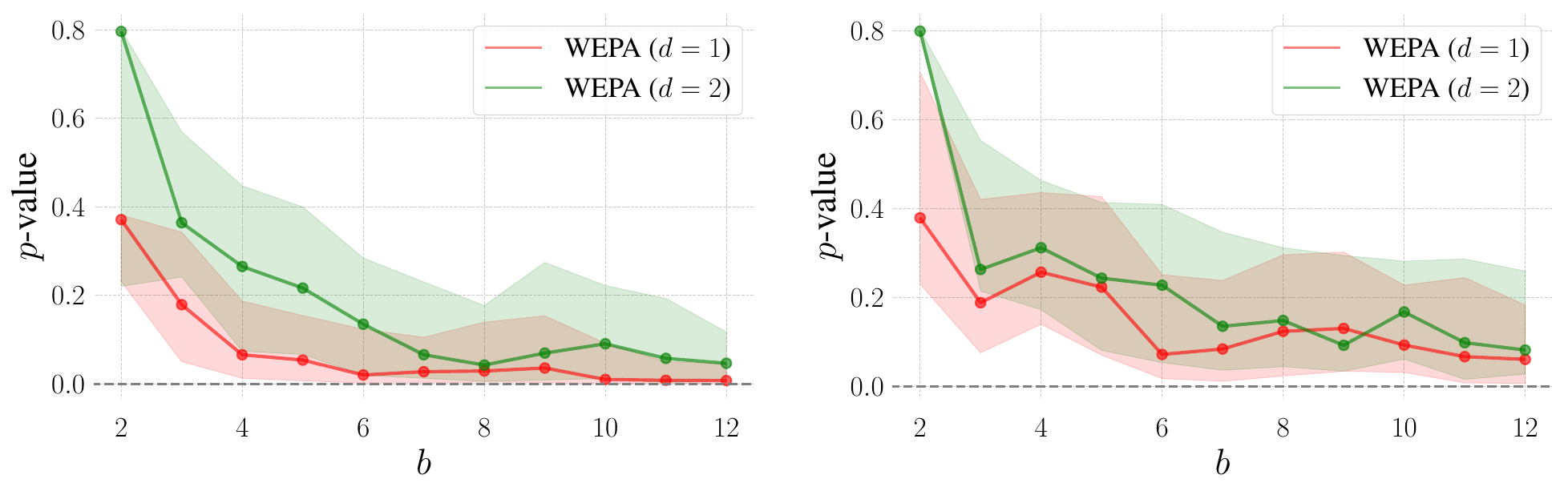}
\captionof{figure}{Median $p$-values varying bitwidths on LLaMA-3B (left) and Mistral-7B (right).}
\label{fig:varying_bits}
\end{figure}

\subsection{Robustness to edit-based attack on Mistral-7B}
\label{app:corruption2}
We evaluated robustness against edit-based attacks on Mistral-7B, following the setup in Section~\ref{sec:robust}, with results shown in Figure~\ref{fig:corruption2}. Under a lower-entropy model, KGW demonstrates greater stability. Overall, these schemes exhibit performance similar to the LLaMA-3B model.

\begin{figure*}[htbp]
    \centering
    \includegraphics[width=\linewidth]{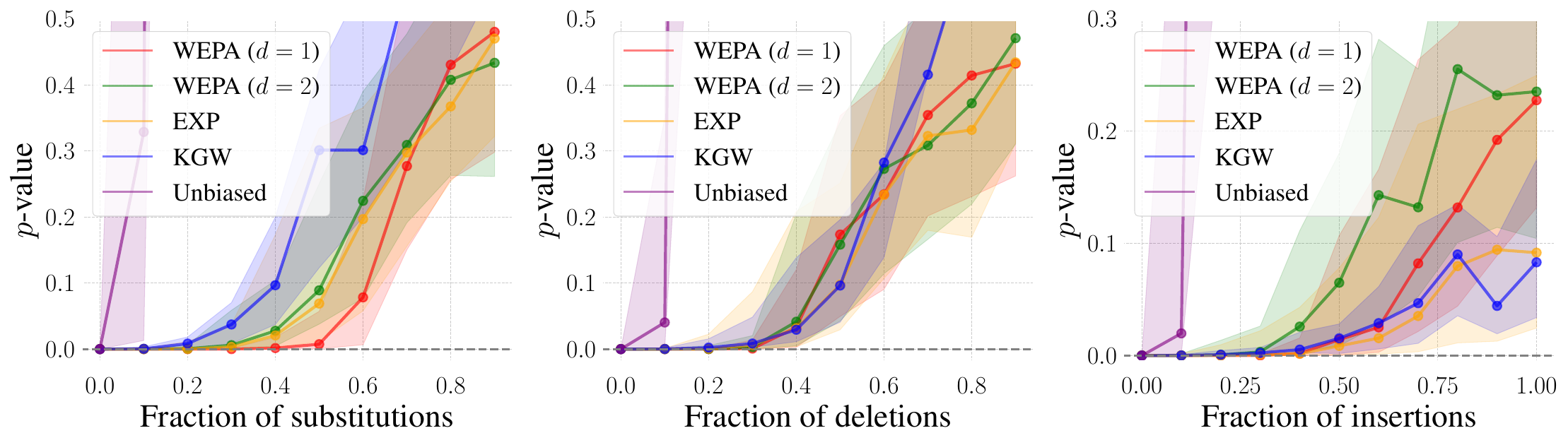}
    \vspace{-1.5em}
    \caption{Median $p$-values under random substitution (left), deletion (middle),
    and insertion (right) attacks.}
    \label{fig:corruption2}
\end{figure*}

%% file: 902-setup.tex
\section{Experiment Setup}
\subsection{Computation resources}
\label{app:res}
Experiments were run on an internal cluster with dual AMD EPYC 7453 28-core CPUs (56 cores total), 1\,TiB RAM, and 8 NVIDIA GeForce GPUs (24\,GiB each, CUDA 11.6). Peak GPU memory usage reached approximately 15\,GiB per device. The system provided ample memory and storage, with over 800\,GiB RAM available during runs. However, our experiments primarily rely on GPU resources for language model inference, and do not require high-performance CPUs.

\subsection{Hyperparameter settings}
\label{app:hypermeters}

We detail the hyperparameter configurations used for all methods evaluated in our experiments. For WEPA, we set the key length $\lambda$ to $256$. For the cyclic key sequence watermarking baseline (EXP), we used a key sequence length of $256$ and applied exponential minimum sampling with a soft Levenshtein cost parameter $\gamma = 0$. For the green-red set watermarking baseline (KGW) \citet{kirchenbauer2023watermark}, we the green set fraction is set at $0.25$ and the logit bias $\delta = 2$. For the unbiased watermarking method of \citet{hu2024unbiased}, we used $\delta$-reweighting.

For WEPA, unless specified, we used the key length of $\lambda$ $256$ and \texttt{float32} type to mimic the behavior when $b$ is sufficiently large. For Levenshtein distance, we set the deletion cost $\gamma_d = 0$ and insertion cost $\gamma_i=2$. Notably, $\gamma_i$ cannot be too small as a lower insertion cost allows tokens to align with arbitrary states with minimal penalty. For each experiment, we computed the $z$-scores for watermarked text across $100$ samples.

%% file: 903-details.tex
\section{Implementation details}
\label{app:imp-details}

\subsection{Implementation of Levenshtein distance calculation}
\label{app:edit-dist}
Algorithm~\ref{algo:edit_dist_dp} presents the dynamic programming for computing the Levenshtein distance used in WEPA. The $\text{cost}(\cdot, \cdot)$ array stores the precomputed costs between substates and tokens defined by~\eqref{eq:d0_dist}.

\begin{algorithm}[H]
    \caption{Levenshtein Distance Calculation with Dynamic Programming}
    \label{algo:edit_dist_dp}
    \begin{algorithmic}[1]
    \Procedure{LevenshteinDistance}{$\by$, $\lambda$, $d$, $\text{cost}(\cdot, \cdot)$, $\gamma_d$, $\gamma_i$}
        \State $m \gets |\by|$
        \State $f, g \gets \text{array of size } \lambda \text{ initialized to } 0$ \Comment{Assume $f$ and $g$ are cyclic.}
        \For{$i \gets 1 .. m$}
            \For{$u \gets 0 .. \lambda-1$}
                \State $d_0 \gets \text{cost}(u, y_{i-1})$  %
                \State $f_{u} \gets g_{u} + \gamma_d$ \Comment{Deletion}
                \For{$v \gets u - d .. u-1$}
                    \State $f_{u} \gets \min(f_{u}, g_{v} + d_0)$ \Comment{Substitution}
                \EndFor
            \EndFor
            \State $u^* \gets \arg\min (f_{:})$
            \For{$u \gets u^* + 1 .. \lambda-1$}
                \For{$v \gets u - d .. u-1$}
                    \State $f_{u} \gets \min(f_{u}, f_{v} + \gamma_i)$ \Comment{Insertion}
                \EndFor
            \EndFor
            \For{$u \gets 0 .. u^*-1$}
                \For{$v \gets u - d .. u-1$}
                    \State $f_{u} \gets \min(f_{u}, f_{v} + \gamma_i)$ \Comment{Insertion}
                \EndFor
            \EndFor
            \State $f, g \gets g, f$
        \EndFor
        \State \Return $\min(f_{:})$
    \EndProcedure
    \end{algorithmic}
\end{algorithm}

While the innermost loop can be optimized using a monotonic queue for a complexity of $\calO(m\lambda)$, we do not carry out this optimization in our implementation in practice since the degree $d$ is usually small.

%% file: 904-analysis.tex
\section{Additional analysis of existing watermarking schemes}
\label{sec:analysis}

In this section, we provide the alternative analysis of existing watermarking schemes through the lens of probabilistic automata (PA).

\subsection{Watermarking schemes under uniform models}

We assume a model that follows a uniform distribution over $\calV$ so that $\mathsf{Model}(\cdot)=\text{Cat}(1/|\calV|,\dots,1/|\calV|)$, where $\text{Cat}$ denotes the categorical distribution. Then we proceed to analyze the output distribution of existing watermarking schemes under uniform models.

\paragraph{$k$-gram-based watermarking.}
$k$-gram-based watermarking has been introduced in many works \citep{kirchenbauer2023watermark,zhao2024provable,aaronson2022aisafety}.
Generally speaking, the core concept is to condition the next token on a window of $k$ prior tokens. A hash function $h: \calK \times \calV^{k}\to \mathbb{Z}$ generates the noise based on the context of the $k$ prior tokens. Then $\Phi$ computes the hash and alters the distribution and generates a variable:
\begin{equation}
    \label{eq:n-gram-ori}\xi_{i}\sim \Phi(\by_{i-k:i-1})=\Phi'(h_{\sk}(\by_{i-k:i-1}
    )),
\end{equation} where $\Phi':\bbZ \to \Delta(\Xi)$. The noise, $\xi_{i}$ may be modeled as a uniform variable from $[0,1]$ used for inverse transform sampling with a partition over the vocabulary \citep{kirchenbauer2023watermark, zhao2024provable}, or a Gumbel variable \citep{aaronson2022aisafety}.

Under uniform models, the output of a watermarking scheme can be represented by a PA $\calM = (Q, \Sigma, \delta , \pi_{0}, \pi_{f})$, with the state space $Q = \calV^{k}$ and $\Xi=\calV$, the transition function $\delta: Q \times \Sigma \times Q \to [0, 1]$ is defined as
\begin{equation}
    \begin{aligned}
        \delta(\by_{i-k:i-1}, y_{i}, \by_{i-k+1:i}) = \bbP(& Y= y_i \mid \by_{i-k:i-1}), \\
        & Y \sim \Phi(\by_{i-k:i-1}).
    \end{aligned}
\end{equation}
As each $\by_{i-k:i-1}$ and $y_i$ uniquely determines $\by_{i-k+1:i}$, $\calM$ is a PDFA.

\paragraph{Cyclic key sequence watermarking.}
\citet{kuditipudi2024robust} introduced a watermarking scheme that uses a cyclic sequence of noise $\dots\xi_{\lambda-1}\xi_{0}\xi_{1}\dots\xi_{\lambda-1}\xi_{0}\dots$ starting from random position in watermarks within generated text, where each $\xi_{i} \in [0,1]$ for inverse transform sampling or $\xi_i \in [0,1 ]^{|\calV|}$ for exponential minimum sampling. The noise sequence is referred to as the key sequence.

Since the key sequence is cyclic, the output under uniform models also follows a repeated structure and can be recognized by a PDFA with at most $4\lambda$ states, which can be constructed with simliar ideas as suffix automata.
We defer the construction in Appendix~\ref{app:sam_cyclic}.

\subsection{Detectability of watermarking schemes}

\begin{definition}[KL-PAC Learnability]
    \label{def:kl-pac}
    Given a class of stochastic languages or distributions $\mathcal{C}$ over $\Sigma^*$, an algorithm $\mathcal{A}$ \emph{KL-Probably Approximately Correctly (KL-PAC)-learns} $\mathcal{C}$ if there exists a polynomial $q$ such that for all $c \in \mathcal{C}$, all $\epsilon > 0$ and $\delta > 0$, $\mathcal{A}$ is given a sample $S_m$ and produces a hypothesis $H$ satisfying
    \begin{equation}
        \Prob \big[D_{\mathsf{KL}}(c || H) > \epsilon \big] < \delta
    \end{equation}
    whenever $m > q(1/\epsilon, 1/\delta, |c|)$, where $|c|$ is some measure of the complexity of the target. The algorithm runs in time polynomial in $m$ plus the total length of the strings in $S_m$.
\end{definition}

\begin{theorem}
    \label{thm:pdfa-detectable} For any $\mu>0$, any watermarking scheme the output of which under a uniform model represented by a $\mu$-distinguishable PDFA with polynomial many states is detectable.
\end{theorem}

Theorem~\ref{thm:pdfa-detectable} provides an alternative perspective to show that the watermarks proposed by \citet{kirchenbauer2023watermark,zhao2024provable,aaronson2022aisafety,kuditipudi2024robust} are all detectable and even spoofable.

%% file: 905-proof.tex
\section{Proof of technical results}
\label{app:proof}

\subsection{Proof of Theorem~\ref{thm:perfect-distortion-free}}
\label{app:perfect-distortion-free}

\begin{proof}
We begin we noticing that

\begin{equation}
\text{Adv}_{\calW}^{\indot}(\calA) = \left| \Pr \left[ b=b' \right] - \frac{1}{2} \right| = 0
\quad \Leftrightarrow \quad \Pr \left[ b=b' \right] = \frac{1}{2}.
\end{equation}

\paragraph{Necessity.} 
Fix an arbitrary output $\by \in \calL$ for a given key $\sk \in K$. Suppose an adversary $\mathcal{A}$ attempts to distinguish between the original and watermarked models. The adversary may partition the output space $\calL$ into two disjoint sets:
\begin{equation}
\calL_{0}, \quad \calL_{1},
\end{equation}
where $\calL_{0} \cap \calL_{1} = \varnothing$ and $\calL_{0} \cup \calL_{1} = \calL$. The adversary's strategy is to output $b' = 0$ if $\by \in \calL_0$ and $b' = 1$ otherwise. The probability of correctly guessing $b$ is then given by:

\begin{equation}
\begin{aligned}
\Pr[b = b'] &= \Pr[b' = 0 \mid b = 0] \Pr[b = 0] + \Pr[b' = 1 \mid b = 1] \Pr[b = 1] \\
&= \frac{1}{2} \left( \Pr[\by \in \calL_0 \mid K = \sk, b = 0] + \Pr[\by \in \calL_1 \mid K = \sk, b = 1] \right).
\end{aligned}
\end{equation}

Since the watermarking scheme $\calW$ is assumed to be perfectly distortion-free, we have:

\begin{equation}
\Pr[Y = \by] = \E_{\sk \leftarrow \mathsf{Gen}(1^\lambda)} \Bigl[ \Pr[Y = \by \mid K = \sk] \Bigr].
\end{equation}

Taking expectations over the key $\sk$, it follows that:

\begin{equation}
\Pr[\by \in \calL_0] = \E_{\sk \leftarrow \mathsf{Gen}(1^\lambda)} \left[ \Pr[\by \in \calL_0 \mid K = \sk] \right],
\end{equation}
with a similar expression for $\calL_1$. By the law of total probability, we obtain:

\begin{equation}
\Pr[\by \in \calL_0] + \Pr[\by \in \calL_1] = 1.
\end{equation}

Substituting into the expression for $\Pr[b = b']$ and noting that $b$ is independent of $\sk$, we conclude:

\begin{equation}
\Pr[b = b'] = \frac{1}{2} \left( \Pr[\by \in \calL_0] + \Pr[\by \in \calL_1] \right) = \frac{1}{2}.
\end{equation}

Since the adversary cannot achieve a probability of success greater than $\frac{1}{2}$, we conclude:

\begin{equation}
\text{Adv}_{\calW}^{\indot}(\calA) = \left| \Pr[b = b'] - \frac{1}{2} \right| = 0.
\end{equation}

Thus, the adversary gains no distinguishing advantage, proving that the watermarking scheme $\calW$ is perfectly indistinguishable.

\paragraph{Sufficiency.}
Conversely, assume that the watermarking scheme $\calW$ is \emph{not} perfectly distortion-free. Then there exists a language model with output distribution $\calL$ and some $\by \in \calL$ such that:

\begin{equation}
\Pr[Y = \by] \;\neq\; \E_{\sk \leftarrow \mathsf{Gen}(1^\lambda)} \bigl[ \Pr[Y = \by \mid K = \sk] \bigr].
\end{equation}

Define the disjoint sets:

\begin{equation}
\calL_0 \;=\;\{\by \in \calL \mid \Pr[Y=\by] \;>\; \E_{\sk \leftarrow \mathsf{Gen}(1^\lambda)}\bigl[\Pr[Y=\by \mid K=\sk]\bigr]\},
\quad
\calL_1 \;=\;\calL \setminus \calL_0.
\end{equation}

Consider an adversary $\calA$ that outputs $b' = 0$ when $\by \in \calL_0$ and $b' = 1$ when $\by \in \calL_1$. The probability of a correct guess is given by:

\begin{equation}
\Pr[b = b'] = \frac{1}{2} \left( \Pr[\by \in \calL_0 \mid b=0] + \Pr[\by \in \calL_1 \mid b=1] \right).
\end{equation}

Since $\by \in \calL_0$ appears more frequently in the original model’s distribution than in the watermarked model’s, and vice versa for $\calL_1$, it follows that:

\begin{equation}
\Pr[b = b'] > \frac{1}{2}.
\end{equation}

Thus, the adversary gains a nonzero advantage, contradicting the assumption that $\calW$ is perfectly distortion-free, which completes the proof.
\end{proof}

\subsection{Proof of Theorem~\ref{thm:pdfa-detectable}}
\label{app:thm-pdfa-detectable}

\begin{proof}
    Since any $\mu$-distinguishable PDFA with a polynomial number of states and a bound on the expected length of strings generated from any state is KL-PAC-learnable \citep{clark2004pac}. That is, given sufficiently many samples from the distribution induced by a $\mu$-distinguishable PDFA, an efficient learning algorithm can approximate the distribution arbitrarily well.

    Suppose, for the sake of contradiction, that there exists a watermarking scheme $\mathcal{W}$ that produces outputs indistinguishable from those of an unwatermarked uniform model represented by a $\mu$-distinguishable PDFA $\mathcal{M}$. That is, for any adversary $\mathcal{A}$,
    \begin{equation}
        \text{Adv}_{\mathcal{W}}^{\text{ind}}(\mathcal{A}) = \left| \Pr[b = b'] - \frac{1}{2} \right| = 0.
    \end{equation}
    This implies that no efficient adversary can distinguish between the original model $\mathcal{M}$ and the watermarked model $\mathcal{M}'$ with non-trivial probability.

    However, since $\mathcal{M}$ is $\mu$-distinguishable, it is KL-PAC-learnable. Thus, there exists a polynomial-time algorithm $\mathcal{L}$ that, given polynomially many samples from the output distribution of either $\mathcal{M}$ or $\mathcal{M}'$, can learn a hypothesis $h$ that approximates the true distribution up to any desired accuracy $\epsilon$. 

    Consider the following adversary $\mathcal{A}$ that attempts to distinguish between $\mathcal{M}$ and $\mathcal{M}'$:
    \begin{enumerate}
        \item Draw $m = \text{poly}(1/\epsilon, 1/\delta)$ samples from the given black-box model.
        \item Run the KL-PAC-learning algorithm $\mathcal{L}$ to obtain a hypothesis $h$ that approximates the underlying distribution.
        \item Compute the likelihood of the observed samples under both the learned approximation of $\mathcal{M}$ and $\mathcal{M}'$.
        \item Output $b' = 0$ if the samples are closer to the expected distribution of $\mathcal{M}$ and $b' = 1$ otherwise.
    \end{enumerate}

    Since $\mathcal{M}$ is $\mu$-distinguishable, every pair of states in $\mathcal{M}$ induces suffix distributions that differ by at least $\mu$ in $\ell_{\infty}$ norm, and $\calM$ terminates with $m = \text{poly}(\lambda)$ steps. If the watermarking scheme $\mathcal{W}$ alters the output distribution, then by definition of $\mu$-distinguishability, it must introduce a statistically significant difference that is detectable given polynomially many samples.

    Therefore, the KL-PAC-learning algorithm $\mathcal{L}$ enables the adversary $\mathcal{A}$ to distinguish between the original and watermarked models with probability strictly greater than $\frac{1}{2} + \gamma$ for some $\gamma > 0$. This contradicts our initial assumption that the watermarking scheme is perfectly indistinguishable.

    Hence, we conclude that for any $\mu > 0$, any watermarking scheme applied to a uniform model represented by a $\mu$-distinguishable PDFA with a polynomial number of states is necessarily detectable.
\end{proof}

Since all watermarking schemes proposed in \citet{kirchenbauer2023watermark, zhao2024provable, aaronson2022aisafety, kuditipudi2024robust} can be represented as a PDFA, and each pair of states in these automata is distinguishable (except in the unlikely case where two randomly chosen partitions are identical or the corresponding uniform random variables are nearly indistinguishable) it follows that all such watermarking schemes are necessarily detectable. %

\section{Undetectable watermark construction}
\label{app:pnfa-undetectable}

Before presenting the watermark construction, we introduce the necessary definitions and assumptions.

\begin{definition}[PAC Learnability]
\label{def:pac-learnability}
Let $\mathcal{C}$ be a concept class consisting of Boolean functions $c: \mathcal{X} \to \{0,1\}$ over some instance space $\mathcal{X}$. We say that an algorithm $\mathcal{A}$ \emph{Probably Approximately Correctly (PAC)-learns} $\mathcal{C}$ if there exists a polynomial $q$ such that for all target concepts $c \in \mathcal{C}$, all distributions $\mathcal{D}$ over $\mathcal{X}$, and all $\epsilon, \delta > 0$, the algorithm $\mathcal{A}$, given access to i.i.d.\ examples drawn from $\mathcal{D}$ labeled by $c$, outputs a hypothesis $h$ such that
\begin{equation}
    \Pr\big[\,\Pr_{x \sim \mathcal{D}}\left[h(x) \neq c(x)\right] > \epsilon\,\big] < \delta
\end{equation}
whenever the number of examples $m > q(1/\epsilon, 1/\delta, |c|)$, where $|c|$ denotes a measure of the complexity of the target concept. The algorithm runs in time polynomial in $m$ plus the size of the input examples.
\end{definition}

\begin{assumption}[Sparsely Learning Parities with Noise]
\label{assump:sparse-parity}
Let $\mathcal{F} = (F_n)_{n \in \mathbb{N}}$ be the class of parity functions $F_n = \{ f_s(x) = \bigoplus_{i \in [n]} x_i s_i \mid s \in \{0,1\}^n,\ \|s\|_1 = \log n \}$, where $x \in \{0,1\}^n$. Each $f_s$ computes the parity of a logarithmic-sized subset of the input, defined by the support of $s$. No polynomial-time algorithm can PAC learn $\mathcal{F}$ under classification noise rate $q = 1/3$.
\end{assumption}

\begin{definition}[Entropy Bound of Language Model]
\label{assump:entropy-bound}
Let $(y_1, y_2, \dots, y_n)$ be a sequence of tokens generated by a language model over a finite vocabulary $\mathcal{V}$. The model assigns a conditional distribution $p(y_i \mid \by_{1:i-1})$ at each position $i$. We define the entropy bound of a language model as
\begin{equation}
    \underline{H}(\mathsf{Model}) = \inf \big\{\mathbb{E}_{i \sim \mathcal{U}[n]} \left[ H\left( y_i \mid \by_{1:i-1} \right) \right] \big\},
\end{equation}
where $H(y_i \mid \by_{1:i-1})$ denotes the conditional entropy and the expectation is taken over a uniformly random index $i \in [n]$.
\end{definition}

\begin{definition}[Weak Pseudorandom Function Family]
\label{def:weak-prf}
A family of functions $\mathcal{F} = \{f_s : \{0,1\}^n \to \{0,1\}\}_{s \in \{0,1\}^n}$ is a \emph{weak pseudorandom function (PRF) family} if for every probabilistic polynomial-time adversary $\mathcal{A}$, the distinguishing advantage
\begin{equation}
\left| \Pr_{s \sim \{0,1\}^n}[\mathcal{A}^{f_s}(1^n) = 1] - \Pr_{f \sim \mathcal{U}}[\mathcal{A}^{f}(1^n) = 1] \right| \le \text{negl}(n),
\end{equation}
where $\mathcal{U}$ denotes the uniform distribution over all functions $f : \{0,1\}^n \to \{0,1\}$, and the oracle queries to $\mathcal{A}$ are drawn uniformly at random from $\{0,1\}^n$.
\end{definition}

\begin{proposition}[From PAC Unlearnability to Weak Pseudorandomness]
\label{prop:weak-prf}
Let $\mathcal{F} = (\mathcal{F}_n)_{n \in \mathbb{N}}$ be a class of functions that is not PAC learnable under classification noise rate $q$. Then, for $n(\lambda) = \lambda$, the sequence of functions $(\mathcal{F}_{n(\lambda)})_\lambda$ constitutes a weak pseudorandom function family with noise level $q$.
\end{proposition}

Proposition~\ref{prop:weak-prf} is well known (see, for example, \cite{bogdanov2017pseudorandom}), and we present it here in the same formulation as in \cite{golowich2024edit}.

\subsection{Full Statement and Proof of Theorem~\ref{thm:pnfa-undetectable}}
\paragraph{Reduction to a binary vocabulary.}
For analytical convenience, we reduce the vocabulary of $\mathcal{V}$ to a binary vocabulary $\{0,1\}$. Any categorical distribution over $\mathcal{V}$ with $|\calV|$ can be equivalently represented by a binary distribution over $\{0,1\}^{\lceil \log_2 |\calV| \rceil}$ via entropy-preserving encoding. That is, given logits $p \in \Delta^{|\calV|-1}$, we define a mapping to binary sequences where each symbol $v \in \mathcal{V}$ is assigned a unique bitstring $b(v) \in \{0,1\}^{\lceil \log_2 |\calV| \rceil}$, and generation proceeds bit by bit using the induced marginal distributions. This reduction preserves perplexity and generation quality up to negligible statistical error. We therefore assume, without loss of generality, that the model outputs binary tokens.

We present the full statement of Theorem~\ref{thm:pnfa-undetectable} as follows.
\begin{theorem}
    There exists an undetectable watermarking scheme with generation length $n = \Omega(\lambda^{3+c}\log\lambda)$ for some $c>0$, entropy bounded by $\underline{H}(\mathsf{Model}) > 2 + \log_2(1-q) - \frac{3-4q}{4(1-q)}\log_2(3-4q)$ that can be represented by a PA, if sparse LPN is hard with noise level of $0 < q < \frac12$.
\end{theorem}
\begin{proof}
Let $f_{\bs}(\bx) = \bs \cdot \bx \bmod 2$ be a parity function with noise level of $q$ on the support of a secret key $\bs \in \{0,1\}^{\lambda}$, and suppose a language model with binary output alphabet $\{0,1\}$.

At each token position $i$, we generate a pair of Gumbel noise variables $(\mu_0, \mu_1)$ sampled from $U[0,1]^2$, whose binary representations are given by~\eqref{eq:uni-to-bin}.

If $(\lambda + 1) \nmid i$, the pair $(\mu_0, \mu_1)$ is left unmodified. Otherwise, when $(\lambda + 1) \mid i$, we extract a watermark bit $x_i$ via the indicator
\begin{equation}
    x_i = \mathbf{1}\left[\mu_0 < \mu_1\right],
\end{equation}
compute the parity bit $b = f_{\bs}(\bx)$ over the current watermark buffer $\bx$, and enforce $x_i = b$ by outputting $(\mu_0, \mu_1)$ with probability $1-q$ if $b=0$ or with probability $q$ and $b=1$, otherwise outputting $(\mu_1, \mu_0)$. This is easily accomplished by a PA by counting the number of bits of $\bs$. The decoder of~\eqref{eq:gumbel-decoder} is used for each step to generating $y_i$.

The detector identifies all positions $i$ such that $(\lambda + 1) \mid i$ and reconstructs the corresponding Gumbel pairs $(\mu_0, \mu_1)$ used during generation. It computes the bit
\begin{equation}
x_i = \mathbf{1}[\mu_0 < \mu_1],
\end{equation}
and compares it against the output token $y_i \in \{0,1\}$. For each such position, let $z_i = \mathbf{1}[x_i = y_i]$ denote a match indicator. Over $t = \Omega(n/\lambda)$ such comparisons, the detector computes the empirical match rate:
\begin{equation}
\hat{p} = \frac{1}{t} \sum_{i=1}^{t} z_i.
\end{equation}

The detector accepts (i.e., concludes the presence of a watermark) if \( \hat{p} \ge \frac{1}{2} + \theta \) for some threshold \( \theta = \Omega(\lambda^{-1}) \), and rejects otherwise.

\textbf{Undetectability.} Let $D_S$ denote the distribution over input-output pairs induced by our construction, where the input $\bx \in \{0,1\}^{\lambda}$ is uniformly random and the output bit is given by $f_{\bs}(\bx)$ with noise rate $q$. This corresponds to the standard \emph{noisy parity} distribution: with probability $1 - q$, the output is $f_{\bs}(\bx)$, and with probability $q$, it is flipped.

Suppose there exists an algorithm $\mathcal{A}$ that KL-PAC-learns the class of such distributions, i.e., for any $\epsilon > 0$ and $\delta > 0$, given $m$ samples drawn from $D_S$, the algorithm produces a hypothesis distribution $\hat{D}$ such that
\begin{equation}
    \Pr \big[ D_{\mathsf{KL}}(D_S \,\|\, \hat{D}) > \epsilon \big] < \delta
\end{equation}
whenever $m > q(1/\epsilon, 1/\delta, |D_S|)$ for some polynomial $q$.

Now consider a random input $\bx \in \{0,1\}^{\lambda}$. Since the output bit is either $f_{\bs}(\bx)$ or its complement, the two candidate outputs are $\bx 0$ and $\bx 1$. A hypothesis $\hat{D}$ that approximates $D_S$ in KL-divergence necessarily assigns higher likelihood to the correct label in expectation. Therefore, by comparing the probabilities $\hat{D}(\bx 0)$ and $\hat{D}(\bx 1)$, one can predict $f_{\bs}(\bx)$ with probability at least $1 - \epsilon$.

This yields an efficient algorithm that predicts noisy parity with non-negligible advantage, contradicting the Noisy Parity assumption, and thus also the hardness of sparse LPN.

Finally, by Proposition~\ref{prop:weak-prf}, the parity function $f_{\bs}$ can be computed by a probabilistic automaton, and the corresponding output distribution $D_S$ serves as a pseudorandom function. In particular, since $D_S$ is not KL-PAC-learnable under the hardness of sparse LPN, it is indistinguishable from a uniform distribution by any efficient adversary.

Therefore, no efficient detector can distinguish watermarked text from unwatermarked text with non-negligible advantage.

\textbf{Completeness.} 
Let $p_0 \in (0,1/2]$ denote the smallest possible bias such that the binary entropy of $y_i$ satisfies
\begin{equation}
H(y_i) = -p_0 \log_2 p_0 - (1 - p_0) \log_2 (1 - p_0).
\end{equation}

Consider the probability that the sampled watermark bit $x_i$ matches the LM output $y_i$ when $i \mid (\lambda+1)$. Since $x_i = b$ with probability $1 - q$, and is flipped with probability $q$ (due to the noise in $f_{\bs}$), we have the following bound:
\begin{equation}
\begin{aligned}
\Pr \left[x_i = y_i\right]
&\ge \left[2p_0\left(\frac{1}{2}\right)^{\log\lambda}(1-q) + \frac{1}{2}\left(1 - \left(\frac{1}{2}\right)^{\log \lambda}\right)\right] \\
&= \frac{1}{2} + \lambda^{-1} \cdot \left(2p_0 - 2qp_0 - \frac12\right) \\
&= \frac{1}{2} + c\lambda^{-1},
\end{aligned}
\end{equation}
where $c := 2p_0-2qp_0-\frac12 > 0$ by the entropy bound.

Then we have the following Chernoff bound:
\begin{equation}
\Pr\left[\hat{p} < \frac{1}{2} + \theta\right] \le \exp\left(-2t(c\lambda^{-1} - \theta)^2\right) = \exp\left(-\frac{t c^2}{2\lambda^2}\right) = \text{negl}(\lambda).
\end{equation}

\textbf{Soundness.} For unwatermarked text, assuming the Gumbel samples and tokens are independent (i.e., no watermarking mechanism was used), the match indicators $z_i$ are i.i.d. unbiased coin flips:
\begin{equation}
\Pr[x_i = y_i] = \frac{1}{2}.
\end{equation}
Then by Hoeffding’s inequality,

\begin{equation}
\begin{aligned}
\Pr\!\left[\hat{p} \ge \tfrac{1}{2} + \theta \right]
&\le \exp(-2t\theta^2) \\
&\le \exp\!\left(-\lambda^{c}\log\lambda\right) \\
&= \operatorname{negl}(\lambda),
\end{aligned}
\end{equation}
which indicates that the detector rejects unwatermarked text with overwhelming probability.

\end{proof}

%% file: 906-automaton.tex
\section{Illustrations of constructions}
\subsection{Illustration of a subordinate probabilistic automaton}
\label{app:subpa}
To complement the formal description in Section~\ref{sec:constr-lm}, we provide an illustration of a subordinate probabilistic automaton (sub-PA) used in our watermarking framework. The sub-PA is responsible for generating the binary noise vectors $\mu_i \in \{0,1\}^c$ across vocabulary positions, as defined in the main text.

\begin{figure}[ht]
    \centering
    \includegraphics[width=0.6\linewidth]{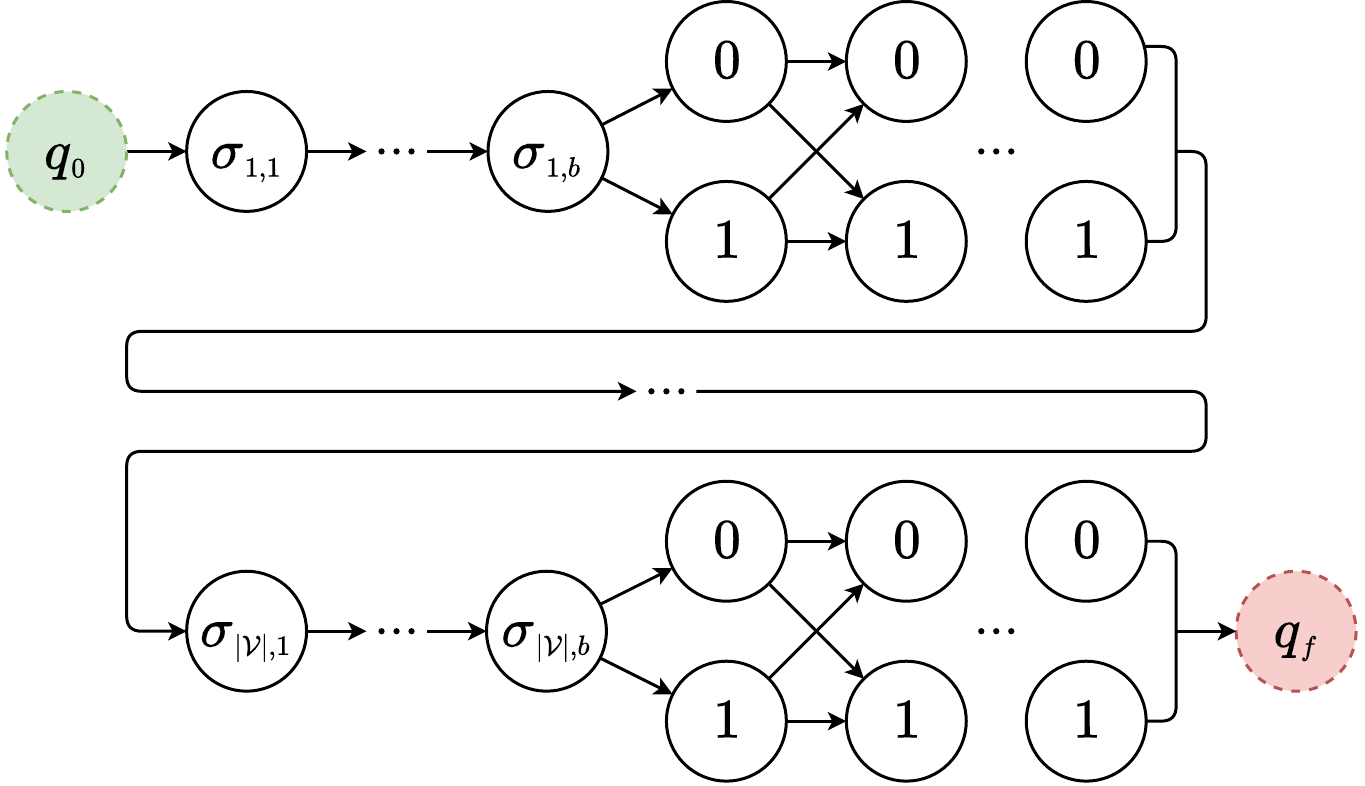}
    \caption{Illustration of a subordinate probabilistic automaton (sub-PA). 
    Each layer encodes a binary vector $\mu_i$, beginning from the initial state $q_0$ and terminating at $q_f$. 
    Transitions within each layer encode bitwise states, branching into parallel Boolean paths $\iota_{i,j}$ and $\hat{\iota}_{i,j}$, which represent $0$ and $1$ respectively. 
    Inter-layer transitions connect the terminal states of layer $i$ to the initial states of layer $i{+}1$.}
    \label{fig:subpa}
\end{figure}

As shown in Figure~\ref{fig:subpa}, the sub-PA progresses from the initial state $q_0$ through $|\mathcal{V}|$ layers, each corresponding to a token in the vocabulary. The Boolean branching ensures exponential variability in possible noise sequences, while the overall number of automaton states remains polynomially bounded. This property is essential for both efficient sampling and tractable computation of edit distance during watermark detection.

\subsection{Automaton construction for a cyclic string}
\label{app:sam_cyclic}
We provide the construction of an automaton for a cyclic string in Algorithm~\ref{alg:sam_cyclic}. Given the input string $s$, the automaton has at most $4|s|$ states.  Figure~\ref{fig:suffix-automata} illustrates an example PDFA that recognizes a substring of the repeated string of ``abaa''.

\begin{algorithm}[t]
\caption{Automaton Construction for a Cyclic String}
\label{alg:sam_cyclic}
\begin{algorithmic}[1]
\State \textbf{Initialize:} Let \( S \) be a set of states representing the suffix automaton, where each state \( S_i \) has:
\begin{itemize}
    \item \( S_i.\texttt{length} \): The length of the longest string ending at \( S_i \).
    \item \( S_i.\texttt{link} \): The suffix link, pointing to the longest proper suffix of the corresponding string that is also in the automaton.
    \item \( S_i.\texttt{next} \): A transition function mapping each character \( c \) to the corresponding state.
\end{itemize}
\State Create the initial state \( S_0 \) with \( S_0.\texttt{length} = 0 \) and \( S_0.\texttt{link} = -1 \).
\State \( \texttt{last} \gets 0 \) \Comment{Tracks the last added state}

\Function{Extend}{$c$}
    \State \( p \gets \texttt{last} \)
    \State Create a new state \( S_{\texttt{cur}} \) with \( S_{\texttt{cur}}.\texttt{length} = S_p.\texttt{length} + 1 \)
    \While{\( p \neq -1 \) and \( c \notin S_p.\texttt{next} \)}
        \State \( S_p.\texttt{next}[c] \gets \texttt{cur} \)
        \State \( p \gets S_p.\texttt{link} \)
    \EndWhile
    \If{\( p = -1 \)}
        \State \( S_{\texttt{cur}}.\texttt{link} \gets 0 \)
    \Else
        \State \( q \gets S_p.\texttt{next}[c] \)
        \If{\( S_p.\texttt{length} + 1 = S_q.\texttt{length} \)}
            \State \( S_{\texttt{cur}}.\texttt{link} \gets q \)
        \Else
            \State Create a new state \( S' \) with \( S'.\texttt{length} = S_p.\texttt{length} + 1 \)
            \State \( S'.\texttt{next} \gets S_q.\texttt{next} \)
            \State \( S'.\texttt{link} \gets S_q.\texttt{link} \)
            \While{\( p \neq -1 \) and \( S_p.\texttt{next}[c] = q \)}
                \State \( S_p.\texttt{next}[c] \gets S' \)
                \State \( p \gets S_p.\texttt{link} \)
            \EndWhile
            \State \( S_q.\texttt{link} \gets S' \)
            \State \( S_{\texttt{cur}}.\texttt{link} \gets S' \)
        \EndIf
    \EndIf
    \State \( \texttt{last} \gets \texttt{cur} \)
\EndFunction

\Function{BuildAutomaton}{$s$}
    \For{each character \( c \) in \( s \)}
        \State \Call{Extend}{$c$}
    \EndFor
    \State \( \texttt{idx} \gets \) index of the newest state in the automaton
    \For{each character \( c \) in \( s \)}
        \State \Call{Extend}{$c$}
    \EndFor
    \State \( S_{\texttt{last}}.\texttt{next}[s_0] \gets \texttt{idx} \)
\EndFunction
\end{algorithmic}
\end{algorithm}

\begin{figure}[ht]
\centering
  \includegraphics[width=0.45\linewidth]{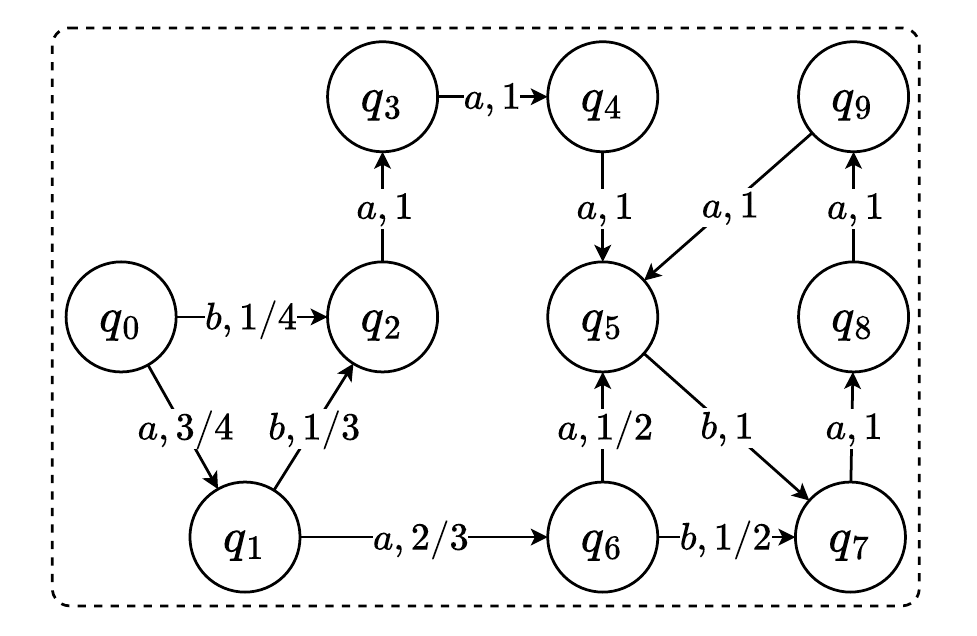}
\caption{An example of a PDFA with probability that recognizes a substring of repeated string ``abaa''.}
\label{fig:suffix-automata}
\end{figure}

\label{appendix:subpa}

%% file: 907-impact.tex
\section*{Impact statement}
\label{app:impact}
This study showcases recent advances in the field of Machine Learning, aimed at preventing the misuse of AI-generated content. Our technique endeavors to improve the security of text generation systems, thereby reducing the risk of misleading information and contributing to the establishment of ownership and copyright of AI-generated content. From an ethical perspective, this technology plays an important role in protecting the integrity of information dissemination on digital platforms and helps strengthen public trust in AI applications. While we recognize that these technologies may be used to restrict freedom of information or enforce censorship in some environments, we believe that their benefits in reducing misinformation outweigh the potential risks.